\documentclass[english]{article}
\renewcommand{\rmdefault}{ptm}
\usepackage[T1]{fontenc}
\usepackage{geometry}
\usepackage{array}
\usepackage{pifont}
\usepackage{amsmath}
\usepackage{amsthm}
\usepackage{amssymb}
\PassOptionsToPackage{normalem}{ulem}
\usepackage{ulem}
\usepackage{nameref}

\makeatletter

\providecommand{\tabularnewline}{\\}

\usepackage{etoolbox}
\usepackage[perpage,bottom]{footmisc}	
\usepackage{scrextend}
\deffootnote[1em]{1em}{0em}{\textsuperscript{\thefootnotemark}\ }
\usepackage{xcolor}
\usepackage{hyperref}	
\hypersetup
{
  unicode=true,
  colorlinks=true,
  linkcolor=[rgb]{0 0 .8},
  citecolor=[rgb]{0 .8 0},
  urlcolor=[rgb]{.8 0 .8},
  breaklinks=true,  
  pdfstartview=,
  pdfauthor={Lim Wei Quan}
}
\usepackage[numbered]{bookmark}
\usepackage[all]{hypcap}
\usepackage[nameinlink]{cleveref}	
\AtBeginDocument{\let\ref\cref}
\crefname{section}{Section}{Sections}
\crefname{figure}{Figure}{Figures}
\crefname{example}{Example}{Examples}
  \let\originalleft\left
  \let\originalright\right
  \renewcommand{\left}{\mathopen{}\mathclose\bgroup\originalleft}
  \renewcommand{\right}{\aftergroup\egroup\originalright}
\binoppenalty=\maxdimen
\DeclareMathSizes{10}{10}{8.5}{7}
\makeatletter
\AtBeginDocument
{
  \check@mathfonts
  \fontdimen13\textfont2 = \fontdimen13\textfont2		
  \fontdimen14\textfont2 = \fontdimen13\textfont2		
  \fontdimen15\textfont2 = \fontdimen13\textfont2		
  \fontdimen16\textfont2 = 1.3\fontdimen16\textfont2	
  \fontdimen17\textfont2 = \fontdimen16\textfont2		
  \fontdimen14\scriptfont2 = \fontdimen13\scriptfont2	
  \fontdimen15\scriptfont2 = \fontdimen13\scriptfont2
  \fontdimen16\scriptfont2 = 1.3\fontdimen16\scriptfont2
  \fontdimen17\scriptfont2 = \fontdimen16\scriptfont2
  \fontdimen14\scriptscriptfont2 = \fontdimen13\scriptscriptfont2	
  \fontdimen15\scriptscriptfont2 = \fontdimen13\scriptscriptfont2
  \fontdimen16\scriptscriptfont2 = 1.3\fontdimen16\scriptscriptfont2
  \fontdimen17\scriptscriptfont2 = \fontdimen16\scriptscriptfont2
}
\makeatother
\usepackage[medium,compact]{titlesec}	
\AtBeginEnvironment{pmatrix}
{
  \setlength{\arraycolsep}{.4em}
  
}
\let\oldsection\section
\let\oldsubsection\subsection
\let\oldsubsubsection\subsubsection
\def\section{\par\addvspace{0em plus 2em}\addpenalty{-5000}\oldsection}	
\def\subsection{\par\addvspace{0em plus 1em}\addpenalty{-4000}\oldsubsection}
\def\subsubsection{\par\addvspace{0em plus .5em}\addpenalty{-3000}\oldsubsubsection}
\let\oldpar\par
\def\par{\oldpar\addpenalty{-1000}}
  \usepackage{enumitem}
  \setlistdepth{14}
  \newlength{\listspace}
  \setlist{topsep=\listspace,itemsep=\listspace,parsep=0em,partopsep=0em}
  \setlist[enumerate]{leftmargin=2em}
  \setlist[itemize]{leftmargin=1.5em}
  \AtBeginEnvironment{list}{\interlinepenalty=500}
\makeatletter
\newcommand{\justified}{%
  \rightskip\z@skip%
  \leftskip\z@skip}
\makeatother
\date{}
\usepackage{upgreek}
\usepackage{calc}	
\newlength{\linespace}

\newlength{\parspace}

\newtheoremstyle{lwq}
  {0em}	
  {0em}	
  {\normalfont}	
  {0em}	
  {\bfseries}	
  {.}	
  {.3em plus .2em}	
  {\thmname{#1}\thmnumber{ #2}\thmnote{ (#3)}}	
\newtheoremstyle{lwqprf}
  {0em}	
  {0em}	
  {\normalfont}	
  {0em}	
  {\itshape}	
  {.}	
  {.3em plus .2em}	
  {\thmname{#1}\thmnote{ (#3)}}	
\newlength{\thmspace}
\newlength{\prfspace}
\newcommand\thmbegin{\par\addvspace{\thmspace}\addvspace{\parskip}\addpenalty{-500}}
\newcommand\prfbegin{\par\addvspace{\prfspace}\addvspace{\parskip}\addpenalty{-500}}
\newcommand\thmend{\par\addvspace{\thmspace}\addpenalty{-500}}	
\newcommand\prfend{\par\addvspace{\prfspace}\addpenalty{-500}}	
\newcommand{\theoremname}{Theorem}
\theoremstyle{lwq}\newtheorem{thm}{\protect\theoremname}
\newcounter{subtheorem}[thm]
\AtBeginEnvironment{thm}{\thmbegin} \AtEndEnvironment{thm}{\thmend}
\renewcommand{\proofname}{Proof}
\theoremstyle{lwqprf}\newtheorem{prf}{\protect\proofname}
\renewenvironment{proof}[1][]{\prfbegin\begin{prf}[#1]\pushQED{\qed}}{\popQED\end{prf}\prfend}
\renewcommand{\qed}{}	
\newcommand{\remarkname}{Remark}
\theoremstyle{lwq}
\AtBeginEnvironment{rem}{\thmbegin} \AtEndEnvironment{rem}{\thmend}
\theoremstyle{lwq}\newtheorem*{rem*}{\protect\remarkname}
\AtBeginEnvironment{rem*}{\thmbegin} \AtEndEnvironment{rem*}{\thmend}
\newcommand{\procedurename}{Procedure}
\theoremstyle{lwq}\newtheorem{proc}[thm]{\protect\procedurename}
\AtBeginEnvironment{proc}{\thmbegin} \AtEndEnvironment{proc}{\thmend}
\newcommand{\setblockspace}
{
  \setlength{\topsep}{0em}
  \setlength{\itemsep}{\listspace}
  \setlength{\parsep}{0em}
  \setlength{\partopsep}{\listspace}
}
\newenvironment{block}
  {\list{}{\leftmargin1.5em\setblockspace\interlinepenalty3000}\raggedright}
  {\endlist}
\newcommand{\definitionname}{Definition}
\theoremstyle{lwq}\newtheorem{defn}[thm]{\protect\definitionname}
\AtBeginEnvironment{defn}{\thmbegin} \AtEndEnvironment{defn}{\thmend}
\newcommand{\algorithmname}{Algorithm}
\theoremstyle{lwq}\newtheorem{algo}[thm]{\protect\algorithmname}
\AtBeginEnvironment{algo}{\thmbegin} \AtEndEnvironment{algo}{\thmend}
\newcommand{\lemmaname}{Lemma}
\theoremstyle{lwq}\newtheorem{lem}[subtheorem]{\protect\lemmaname}
\AtBeginEnvironment{lem}{\thmbegin} \AtEndEnvironment{lem}{\thmend}
\newcommand{\subroutinename}{Subroutine}
\theoremstyle{lwq}\newtheorem{subr}[subtheorem]{\protect\subroutinename}
\AtBeginEnvironment{subr}{\thmbegin} \AtEndEnvironment{subr}{\thmend}
\newcommand{\conjecturename}{Conjecture}
\theoremstyle{lwq}\newtheorem{conjecture}[thm]{\protect\conjecturename}
\AtBeginEnvironment{conjecture}{\thmbegin} \AtEndEnvironment{conjecture}{\thmend}

\usepackage[medium,compact]{titlesec}
\def\section{\par\addvspace{0em plus 2em}\addpenalty{-3000}\oldsection}	
\def\subsection{\par\addvspace{0em plus 1em}\addpenalty{-2000}\oldsubsection}
\def\subsubsection{\par\addvspace{0em plus .5em}\addpenalty{-1000}\oldsubsubsection}

\AtBeginDocument{

}

\makeatother

\usepackage{babel}
\providecommand{\theoremname}{Theorem}

\begin{document}
\newcommand\br{\addpenalty{-1000}}

\renewcommand{\qedsymbol}{$\diamond$}

\global\long\def\xor{\oplus}
\global\long\def\imp{\Rightarrow}
\global\long\def\eq{\Leftrightarrow}

\global\long\def\nn{\mathbb{N}}
\global\long\def\zz{\mathbb{Z}}
\global\long\def\qq{\mathbb{Q}}
\global\long\def\rr{\mathbb{R}}
\global\long\def\cc{\mathbb{C}}
\global\long\def\ff{\mathbb{F}}
\global\long\def\hh{\mathbb{H}}

\global\long\def\wi{\subseteq}
\global\long\def\co{\supseteq}
\global\long\def\nwi{\nsubseteq}
\global\long\def\nco{\nsupseteq}
\global\long\def\none{\varnothing}

\global\long\def\floor#1{\left\lfloor #1\right\rfloor }
\global\long\def\ceil#1{\left\lceil #1\right\rceil }

\global\long\def\cond#1#2#3{\left(\text{ }#1\text{ ? }#2\text{ : }#3\text{ }\right)}

\global\long\def\class#1#2{\left\{  \,#1\,:\,#2\,\right\}  }
\global\long\def\seq#1#2{\left(\,#1\,:\,#2\,\right)}

\global\long\def\t#1{\text{#1}}
\global\long\def\str#1{\text{``\text{#1}''}}

\global\long\def\f#1{\operatorname{#1}}
\global\long\def\v#1{\mathit{#1}}
\global\long\def\p#1{\text{.#1}}

\global\long\def\w{\omega}
\global\long\def\W{\Omega}
\global\long\def\Q{\Theta}

\global\long\def\e{\upvarepsilon}
\global\long\def\g{\upgamma}

\noindent \begin{center}
\textbf{\Large{}Dynamic Reallocation Problems in Scheduling}
\par\end{center}{\Large \par}

\noindent \begin{center}
\begin{tabular}{>{\centering}p{0.3\textwidth}>{\centering}p{0.3\textwidth}>{\centering}p{0.3\textwidth}}
\textbf{Wei Quan Lim}

limweiquan@nus.edu.sg & \textbf{Seth Gilbert}

seth.gilbert@comp.nus.edu.sg & \textbf{Wei Zhong Lim}

limweizhong@nus.edu.sg\tabularnewline
\end{tabular}
\par\end{center}

\section*{Keywords}

Online problems, scheduling, reallocation, distributed systems

\section*{Abstract}

In this paper we look at the problem of scheduling tasks on a single-processor
system, where each task requires unit time and must be scheduled within
a certain time window, and each task can be added to or removed from
the system at any time. On each operation, the system is allowed to
reschedule any tasks, but the goal is to minimize the number of rescheduled
tasks. Our main result is an allocator that maintains a valid schedule
for all tasks in the system if their time windows have constant size
and reschedules $O\left(\frac{1}{\e}\log(\frac{1}{\e})\right)$ tasks
on each insertion as $\e\to0$, where $\e$ is a certain measure of
the schedule flexibility of the system. We also show that it is optimal
for any allocator that works on arbitrary instances. We also briefly
mention a few variants of the problem, such as if the tasks have time
windows of difference sizes, for which we have an allocator that we
conjecture reschedules only $1$ task on each insertion if the schedule
flexibility remains above a certain threshold.

\section*{Acknowledgements}

We would like to express our gratitude to our families and friends
for their untiring support, and to all others who have given us their
invaluable comments and advice.

\clearpage{}

\section{Introduction}

Scheduling problems with time restrictions arise everywhere, from
doctors\textquoteright{} appointments to distributed systems. Traditionally
the focus is on finding a good allocation of resources under intrinsic
constraints such as availability of resources, deadlines and dependencies,
and extrinsic requirements such as fairness and latency. Sometimes
all we need is a solution that satisfies the constraints. At other
times we desire an allocation that maximizes some objective function,
involving various factors such as the resources used, throughput and
latency.

In online scheduling, it is often possible or even necessary to reallocate
resources that had been previously reserved for preceding requests
so as to satisfy a new request. The cost of reallocation should therefore
be taken into account as well. Furthermore, it would be preferable
if the system can be loaded to nearly its full capacity and yet have
only a small reallocation cost for each new request serviced.

\section{The real-time arbitrary reallocation problem}

We have $p$ identical processors in our system, and each of them
starts with an empty schedule. An insert operation is the insertion
of a new task into the system with a specified length and time window
in which it has to be executed. The task has to be allocated to one
of the processors at some time slot of the required length within
the given window such that it does not overlap with the time slot
of any other task currently in the processor\textquoteright s schedule.
To do so, we may have to reallocate time slots for other tasks in
the schedule, and possibly even reallocate some tasks to different
processors in the system. We must allow any insertion that is feasible,
even if it requires reallocating all the tasks. A deletion operation
is simply the deletion of a task currently in the system. We allow
reallocations on deletions, but all the allocators described in this
paper will not make any reallocations. The goal is to minimize the
number of such reallocations on each insertion in the worst case.

It is easy to see that if tasks can be of different lengths or the
schedule is allowed to become completely packed, there is no efficient
allocator (i.e. for any $n$ there is a sequence of operations that
force $\W(n)$ reallocations per insertion on average although there
are at most $n$ tasks in the system at any time). The situation improves
when tasks are of unit length and there is some slack in the schedule,
which is hence the focus of this paper. Indeed, the main result of
this paper is that for unit-length tasks and fixed-length windows,
given any $\e>0$ it is possible to maintain a valid schedule with
only $O(1)$ reallocations on each insertion as long as the instance
(i.e. set of tasks) prior to insertion is $\e$-slack (i.e. there
is a valid schedule even if those tasks are of length $(1+\e)$ instead).

\section{Main results}

We demonstrate a single-processor allocator \nameref{F.FA} for fixed-length
windows that takes $O\left(\frac{1}{\e}\log(\frac{1}{\e})\right)$
reallocations per insertion for arbitrary $\e$-slack instances, which
is independent of the number of tasks $n$ and even the window length,
and does not reallocate any tasks on deletion. Also, its time complexity
is only $O\left(\frac{1}{\e}\log(\frac{1}{\e})\log(n)+\log(n)^{2}\right)$
per insertion and $O(\log(n))$ per deletion if it is allowed to maintain
an internal state. We also present a hard insert state (i.e. a current
allocation and a new insertion) that gives a matching \nameref{F.LowerBound}
of $\W\left(\frac{1}{\e}\ln(\frac{1}{\e})\right)$ reallocations for
any allocator that solves it.

This resolves a question posed in Bender et al. (2013) \cite{RPIS}.
They investigated the variant with variable-length windows but where
the system is always $\e$-slack for some large constant $\e$, and
obtained an allocator that uses only $O(\log^{\ast}(\min(n,c))$ reallocations
per insertion where $c$ is the maximum window length. Then they asked
if there is an efficient allocator for arbitrarily small $\e$, which
our paper answers completely, in the positive for the special case
of one processor and fixed-length windows, where our allocator uses
only $O(1)$ reallocations per insertion, and in the negative for
multiple processors or variable-length windows.

\section{Related work}

Of the wide variety of scheduling problems, completely offline problems
from real-world situations tend to be NP-complete even after simplifications.
(Li et al., 2008) \cite{PSUURC} But practical scheduling problems
usually involve continuous processes or continual requests, and hence
many other scheduling problems are online in some sense. In addition,
scheduling problems with intrinsic constraints invariably have to
include some means of handling conflicting requests.

\subsection{Dominant resource cost}

A large class of online scheduling problems have focused on minimizing
the maximum machine load or the makespan, or on maximizing utilization,
which are typical targets in static optimization problems. Numerous
types of rescheduling have been studied and we shall give only a few
examples.

\subsubsection*{Adaptive rescheduling}

Hoogeveen et al. (2012) \cite{RNOSMST} developed makespan minimizing
heuristics where tasks have deadlines and also a setup time if the
previous task is of a different type, and insertion of new tasks must
not result in any unnecessary additional setups. Castillo et al. (2011)
\cite{OAARR} proposed a framework similar to this paper's where each
task has a total required duration and a time window within which
it must be executed, and each request must either be accepted or rejected,
immediately and permanently. But their goal is to minimize the number
of rejected requests while maximizing utilization.

If however revoking or reallocating earlier requests is permitted,
accepting new conflicting requests may become possible. For example,
Faigle and Nawijn (1995) \cite{NSIO} presented an optimal online
algorithm for such a problem, where each request demand a certain
length of service time starting from the time of the request, and
must be immediately assigned to a service station or rejected, where
each service station can service only one request at any one time.
A request is considered unfulfilled if it is rejected or if its service
is interrupted, and the objective is to minimize the number of unfulfilled
requests.

In incremental scheduling, value is accrued over time according to
the activities performed and the resources used. Gallagher et al.
(2006) \cite{ISMQDE} presented techniques for one variant of incremental
scheduling with insufficient resources where each activity requires
a minimum time period and each pair of activities requires a setup
time in-between. Whenever the set of possible activities changes,
changing the schedule may yield higher value. Their simulations found
that a local adjustment algorithm results in a rather stable schedule
that performs well compared to a greedy global rescheduling.

\subsubsection*{Delayed rescheduling}

In one problem restricted reallocations are allowed at the end of
the entire request sequence. Tan and Yu (2008) \cite{OSWR} introduced
3 such variants for two machines, where tasks of arbitrary lengths
are to be completed in any order. On each request, the task must be
immediately assigned to a single machine. After the entire sequence
of requests, certain tasks can be reassigned to a different machine.
In the first variant, the last $k$ tasks assigned can be reassigned.
In the second, only the last task assigned to each machine can be
reassigned. And in the third, any $k$ tasks can be reassigned. They
demonstrate algorithms to minimize maximum machine load with optimal
competitive ratios. Min et al. (2011) \cite{OSOASPRTIM} proposed
a fourth variant where the last task of only one machine can be reassigned,
and showed that it has the same competitive ratio as the second variant.
The fifth variant is like the fourth but the total length of all tasks
is known beforehand, and they presented an optimal algorithm with
competitive ratio $\frac{5}{4}$. Liu et al. (2009) \cite{OSTUMMM}
and Chen et al. (2011,2012) \cite{OAOSBRTE,OSORTE} investigated the
generalization where the two machines are of speeds $1$ and $s$,
and $k$ arbitrary tasks can be reassigned at the end, and established
optimal algorithms for some ranges of values of $s$.

Another common method to improve scheduling for a batch of tasks is
to use a buffer to store tasks before assigning them. The many variations
(Dósa and Epstein, 2010; Chen et al, 2013) \cite{OSBRM,SOHSPBR} of
the problem are due to the number of machines, their speeds, and whether
some tasks can be executed only on certain machines, among other factors.
Sun and Fan (2013) \cite{ISOMSRB} analyzed one such load minimization
problem for $p$ identical machines, where the system has a fixed-size
buffer. On each request, the task must be assigned permanently to
a machine or stored in the buffer if it is not full. Each task in
the buffer can be assigned permanently to a machine at any time. They
improved the upper bound on the buffer sizes for both the optimally
competitive algorithm for large $p$ and a $1.5$-competitive algorithm.

\subsubsection*{Real-time rescheduling}

The common drawback to delayed rescheduling models is that they are
inapplicable to continuous real-time systems where, on each request,
immediate resource allocations must be made so that a valid schedule
is always maintained.

Dósa et al. (2011) \cite{OSRTRM} proved that for a finite set of
tasks and two machines with different speeds, using a buffer is not
as efficient as bounded reallocations, where on each request $k$
previously allocated tasks can be reallocated together with the new
task. They also give an optimal bounded reallocation algorithm for
certain parameter ranges.

Sanders et al. (2009) \cite{OSBM} considered proportionally bounded
reallocation cost instead, such that for each new task of size $L$
that is inserted, any set of currently allocated tasks that have total
size bounded by $rL$ can be reallocated to a different machine. They
determined that for any $\g>1$ there is such a scheduler that is
$\g$-competitive for some $r$. Like this paper, this explores the
intermediate class of scheduling algorithms between exact optimal
algorithms and $\g$-competitive algorithms.

\subsection{Dominant reallocation cost}

Real-time scheduling problems tend to involve resources that have
already been reserved for the scheduling system, and thereby it is
not unusual that the dominant cost to be minimized is not the resource
cost but the reallocation cost. In general, the system is also $\g$-underallocated
or $\e$-slack in some sense, in other words the capacity is always
at least $\g=1+\e$ times the load under some suitable measure, and
we seek algorithms whose performance degrades gracefully as load approaches
capacity.

Davis et al. (2006) \cite{OAMRRNC} put forward a neat problem with
a pool of a single resource of total size $T$, and $n$ users, each
requiring a certain number of the resource at each time step. At each
step, the scheduler has to distribute the pool of resources to the
users without knowing their requirements but only knowing which of
them are satisfied. To do so must repeatedly change the resource allocations
until all are satisfied, with the aim of minimizing the total number
of changes to the user allotments. They devise a randomized algorithm
that is $O(\log_{\g}(n))$-competitive if the size of the resource
pool is increased to $\g T$ given any $\g>1$, and also show that
an expanded resource pool is necessary for any $f(n)$-competitive
algorithm given any function $f$. This illustrates the smooth trade-off
between the underallocation and the reallocation cost.

In a similar direction, Bender et al. (2014) \cite{COSR} gave an
optimal cost-oblivious algorithm to maintain, given $\e>0$, an allocation
of memory blocks (with no window constraints) that has makespan within
$(1+\e)$ times the optimal, with a reallocation cost of $O\left(\frac{1}{\e}\log(\frac{1}{\e})\right)$
times the optimal as long as the reallocation cost is subadditive
and monotonic in the block size. And in \cite{CORSP} they gave a
cost-oblivious algorithm to maintain an allocation of tasks to multiple
processors that has sum of completion times within a constant factor
of optimal, with a reallocation cost within $O(1)$ of the optimal
if the reallocation cost is strongly subadditive.

The real-time arbitrary reallocation problem defined at the start
of this paper also has nonzero underallocation and dominant allocation
cost, and there is again a trade-off between minimizing the number
of reallocations and minimizing the amount of underallocation needed.
Bender et al. (2013) \cite{RPIS} investigated one variant where the
system is always $\g$-underallocated for some constant $\g$, and
obtained an algorithm for sufficiently large $\g$ that uses only
$O(\log^{\ast}(\min(n,c))$ reallocations where $n$ is the current
number of tasks and $c$ is the maximum window length. It is not apparent
whether this is asymptotically optimal, and they also ask if there
is an algorithm for arbitrary $\g>1$, which would be more practical.
As such, this paper goes in that direction, and shows that for fixed-length
task windows there is indeed a single-processor reallocation scheduler
such that the number of reallocations needed on each task insertion
is dependent on only the current underallocation. We also show partial
results for variable-length task windows, firstly that it is impossible
if underallocation is below some threshold, and secondly that a nonzero
reallocation rate is inevitable regardless of underallocation. For
sufficiently large underallocation, we have an allocator that we conjecture
makes at most a single reallocation per insertion.

\clearpage{}

\section{Summary\label{Summary}}

In this paper an instance is a set of tasks in the system, and is
said to be $\g$-underallocated and $\e$-slack iff there is still
a solution when task lengths are multiplied by $\g=1+\e$, which we
shall call a $\g$-solution. The slack of an instance is then defined
as the maximum $\e$ such that it is $\e$-slack, and this maximum
can be seen to exist by compactness. We decided to look mainly at
the unit-task case.

The first variant is where the window lengths are all the same. We
explicitly detail a single-processor allocator \nameref{F.FA} that
needs only $O\left(\frac{1}{\e}\ln(\frac{1}{\e})\right)$ reallocations
per insertion for arbitrary instances, which is independent of the
number of tasks $n$ and even the window length, and does not reallocate
any tasks on deletion. Also, its time complexity is only $O\left(\frac{1}{\e}\log(\frac{1}{\e})\log(n)+\log(n)^{2}\right)$
per insertion and $O(\log(n))$ per deletion if it is allowed to maintain
an internal state. To that end we prove a few preliminary results
that would be frequently used throughout the proof of this allocator,
including the basic theorem \nameref{F.Ord} and procedures \nameref{F.Left}
and \nameref{F.Near}. We also construct a hard situation that gives
a \nameref{F.LowerBound} of $\W\left(\frac{1}{\e}\ln(\frac{1}{\e})\right)$
reallocations, suggesting that an allocator might in fact need that
number of reallocations for some operations, although we do not know
whether there is an allocator that can perpetually avoid such situations
altogether and, nor if there is one that has asymptotically lower
amortized reallocation cost.

The second variant is where window lengths can be different, where
Bender et al. asked if there is an efficient allocator for arbitrary
positive slack \cite[Open questions]{RPIS}. We answer the question
in the negative under rather weak conditions in \ref{V.SmallSlack}.
Specifically, we say that there is no efficient allocator iff for
$p$ processors and for any $n$ there is a sequence of operations
that force $\W(\frac{n}{p})$ reallocations per insertion on average
although there are at most $n$ tasks in the system at any time. Then
for any $\e<\frac{1}{3}$ there is no efficient allocator even if
the instance is always $\e$-slack. This raises the question of whether
there is an efficient allocator if $\e\ge\frac{1}{3}$. At the other
extreme, regardless of the underallocation there is still a \nameref{V.ReallocReq},
where any allocator can be forced to make at least one reallocation
for some sequence of insertions.

If the instance always remains $4\g$-underallocated where $\g$ is
a power of $2$, we can \nameref{V.Align} it and maintain a solution
for the aligned instance, where windows have endpoints recursively
aligned to powers of $2$. This \nameref{V.Alignment} was mentioned
by Bender et. al. \cite[Lemma 10]{RPIS}, but as stated there it is
incorrect \footnote{Their main result is still valid because it depends only on the correct
reduction theorem.} because there is a counter-example for $\g=27$. As they noted in
\cite[Lemma 4]{RPIS}, any insertion such that the instance remains
aligned can be solved in $O(\log(n))$ reallocations, and this is
asymptotically tight for some sequence of operations if there is no
\nameref{V.UnderallocReq}. Furthermore, all allocators have \nameref{V.NonGeneric},
in the sense that for any $\g$ and as $n\to\infty$ there is some
$\g$-underallocated insert state which requires $\W(\log(n))$ reallocations
to be solved. Therefore any allocator that does better must completely
avoid such insert states. We believe that it is indeed possible to
maintain a solution if the instance is always aligned and $2$-underallocated,
and in fact we have an allocator \nameref{V.VA} that we conjecture
makes at most $1$ reallocation per insertion. Empirically it worked
in all our experiments using partially random operations, but the
search space is too large for these tests to be reliable.

Finally we reduce the $p$-processor problem to the $1$-processor
problem in \ref{M.Reduction}, giving an allocator that is guaranteed
to work if the original slack is larger than $2\frac{p-1}{p+1}$.
It is doubtful that this allocator is optimal, but it is also not
clear how to do better. We also show in \ref{M.SmallSlack} that even
if all windows are of the same length and the instance always remains
$\e$-slack, there is no efficient allocator if $p>1$ and $\e<\frac{1}{4p-1}$.
This is a rather strong negative result in light of our efficient
allocator FA for one processor, but leaves unanswered what happens
when $\e\in[\frac{1}{4p-1},2\frac{p-1}{p+1}]$.

At the end we briefly discuss the case of variable-length tasks. Even
if the instance is always $\g$-slack before each insertion, there
is no efficient allocator regardless of how large $\g$ is. And even
if the instance is always $\g$-slack (including after insertions),
there is no efficient allocator if $\g<2$, but we do not know what
happens when $\g\ge2$.

\section{Data structures}

Although the algorithms described guarantee the existence of a solution
with low reallocation cost, actually finding the solution would be
inefficient without the following data structures. An IDSet is an
iterable ordered set data structure based on immutable balanced binary
trees with all data at the leaves, and it allows access by both id
and position, besides the usual update and search operations. An IDSetRQ
is an augmented version of IDSet to also allow range queries for any
user-defined associative binary function. These data structures take
worst-case $O(\log(n))$ time for any operation, including deep copying.

\clearpage{}

\section{Definitions}

We shall begin by defining general symbolic notation that are not
universal but will be used in this paper to express all statements
precisely yet concisely, and then we shall define terms specific to
the reallocation problem.

\subsection{General symbolic notation\label{Def.Symbolic}}

Let $null$ be a sentinel value denoting ``nothing''.

Let $x[i]$ denote the element in $x$ indexed by $i$ for any sequence
$x$ and integer $i$.

Let $x[y]=\seq{x[i]}{i\in y}$ for any sequence $x$ and integer sequence
$y$, where ``$\in$'' denotes ``is generated in order by''.

Let $[a..b]$ be the strictly increasing sequence $\seq x{x\in\zz\land a\le x\le b}$
for any reals $a,b$, and let $x[a..b]=\seq{x[i]}{i\in[a..b]}$.

Let $()$ denote the empty sequence, and for any $x$ let $\f{seq}(x)$
be the sequence that has $x$ as its only element.

Let $x\cdot y$ be the concatenation of $x$ followed by $y$ for
any finite sequences $x,y$.

\global\long\def\first{\p{first}}
\global\long\def\last{\p{last}}

Let $x\first$ and $x\last$ be the first and last element in $x$
respectively for any sequence $x$, and let $\#(x)$ be its length.

\global\long\def\start{\p{start}}
\global\long\def\endd{\p{end}}

Let $X\start$ and $X\endd$ be the start and end respectively of
$X$ for any interval $X$, and let $\f{span}(X)$ be its length/span.

Let $\le$ on intervals be the partial ordering such that $X\le Y\eq X\start\le Y\start\land X\endd\le Y\endd$
for any intervals $X,Y$.

For convenience let left/right be associated with earlier/later for
interval comparison.

Call a integer/interval sequence $x$ \textbf{ordered} iff its non-$null$
elements are in increasing order, and for any integer sequence $y$
let ``\textbf{sort $x[y]$}'' mean ``permute the non-$null$ elements
among $x[y]$ such that $x[y]$ is ordered''.

Let $c+[a,b]=[a+c,b+c]$ and $[a,b]c=[ac,bc]$ for any reals $a,b,c$.

Let $\f{span(}S)=\max_{X\in S}X\endd-\min_{X\in S}X\start$ for any
finite set/sequence of intervals $S$.

Let $\cond PXY$ evaluate to $X$ iff $P=true$ and $Y$ otherwise
for any boolean $P$ and expressions $X,Y$.

\subsection{Reallocation problem terminology\label{Def.DRPIS}}

Take any \textbf{instance} $I=(n,T,W)$ that comprises a set of $n$
unit tasks $T[1..n]$ and their windows $W[1..n]$. For convenience
we shall often not mention the tasks but associate an allocated slot
directly with the task's window.

Call $S$ a \textbf{valid allocation for} $I$ iff it is an allocation
of tasks in $I$ such that each task $T[i]$ in $I$ is allocated
to the slot $S[i]$ of unit length within $W[i]$ or $S[i]=null$
if $T[i]$ is unallocated, and no two tasks in $I$ are allocated
to overlapping slots. For convenience we shall use ``slot'' to refer
to a unit interval unless otherwise specified.

Call $I$ \textbf{ordered} iff $W[1..n]$ is ordered. If so, call
a valid allocation $S$ for $I$ \textbf{ordered} iff $S[1..n]$ is
ordered. (Unallocated tasks are ignored.)

Call $S$ a \textbf{solution for $I$} iff $S$ is a valid allocation
for $I$ that allocates all tasks in $I$.

Call $S$ a \textbf{partial solution for $(I,k)$} iff $S$ is a valid
allocation for $I$ that allocates all tasks in $I$ except $T[k]$.

Call $S$ a \textbf{$\g$-solution for $I$} iff $S$ is a solution
for $I'$ where $I'$ is $I$ with all task lengths multiplied by
$\g$.

Call $S$ a \textbf{$\g$-partial solution for $(I,k)$} iff $S$
is a partial solution for $(I',k)$ where $I'$ is $I$ with all task
lengths multiplied by $\g$.

Call $I$ \textbf{feasible} iff there is a solution for $I$, and
call $I$ \textbf{$\g$-underallocated} iff there is a $\g$-solution
for $I$.

Call $I$ \textbf{$\e$-slack} iff $I$ is $(1+\e)$-underallocated,
and let the \textbf{underallocation of $I$} be the maximum such $\e$,
which clearly exists.

Call $(I,S,k)$ an \textbf{insert state} iff $S$ is a partial solution
for $(I,k)$, and call it \textbf{ordered} iff $I$ is ordered, \textbf{feasible}
iff $I$ is feasible, and \textbf{$(1+\e)$-underallocated} or equivalently
\textbf{$\e$-slack} iff there is a $(1+\e)$-partial solution for
$(I,k)$.

Call an insertion of a new task \textbf{feasible} iff the resulting
instance is feasible, or equivalently iff it creates a feasible insert
state.

Call an algorithm A an \textbf{allocator} iff it maintains a solution
on any feasible task insertion.

\clearpage{}

\section{Fixed window length}

In this cost model that does not distinguish between tasks of different
lengths, an efficient allocator is impossible if tasks can have arbitrary
lengths, and hence we shall look at only the case of unit-length tasks
until the very end of this paper. This section will be restricted
to the first variant where all the windows have the same fixed length
$c$, and henceforth we can always assume that the instance is ordered.
Our main result is that for one processor the maximum number of reallocations
needed on each insertion for any $\e$-slack instance is $\Q\left(\frac{1}{\e}\log(\frac{1}{\e})\right)$
as $\e\to0$. In the following subsections, we will establish a collection
of useful algorithms and theorems, and then demonstrate a family of
instances that show that the bound cannot be improved, and finally
describe and prove an optimal allocator that realizes the bound while
still having a good time complexity.

\subsection{Preliminaries}

A simple but very helpful theorem is that given any ordered instance
\textit{$I$} and solution $S$ for $I$, sorting $S$ (permuting
the slots so that they are now in the same order as their windows)
gives an ordered solution for $I$. As an easy consequence, given
any valid allocation $S$ for $I$, sorting any subsequence of $S$
in-place (so that the slots in that subsequence are now in the same
order as their windows) gives a valid allocation for $I$. Also, any
feasible instance has an ordered solution.

It is then clear that the greedy algorithm, which allocates the tasks
from left to right each to the leftmost possible slot, produces an
ordered solution $L$, such that given any ordered solution $X$,
every slot in $S$ is no later than the corresponding slot in $X$.
Now given any feasible ordered insert state $(I,S,r)$ with $n$ tasks
such that $S$ is ordered, we can use $L$ to construct a near ordered
solution for $(I,S,r)$, which is defined as an ordered solution $N$
such that $S[r-1]\le N[r]\le S[r+1]$ (the inserted task is allocated
to a slot within the range given by the neighbouring slots in the
original ordered solution).

Along the way, we also define a procedure \nameref{F.Snap} that takes
as inputs a slot $s$ and a window $w$ and returns the slot within
$w$ that is closest to $s$. This procedure will be used later as
well.

The rest of this preliminary results section contain the precise statements
and proofs of the above theorems. All of them can be extended without
difficulty to the multi-processor case, unlike the results in the
later sections.
\begin{thm}[Ordering]
\label{F.Ord} The following are true:
\begin{enumerate}
\item For any ordered instance $I=(n,T,W)$ and solution $S$ for $I$,
$S$ when sorted is an ordered solution for $I$.
\item For any ordered instance $I=(n,T,W)$ and valid allocation $S$ for
$I$ and ordered sequence $k[1..m]$, $S$ with $S[k[1..m]]$ sorted
is a valid allocation for $I$.
\end{enumerate}
\end{thm}
\begin{proof}
We shall first prove (1). Take any ordered instance $I=(n,T,W)$ and
solution $S$ for $I$. Let $S'$ be $S$ sorted.

While $S\ne S'$ we shall iteratively modify $S$ such that the following
invariances hold after each step $j$:
\begin{enumerate}
\item $S[1..j]=S'[1..j]$.
\item $S$ is a solution for $I$.
\end{enumerate}
Then after at most $n$ steps $S[1..n]=S'[1..n]$ and hence $S'$
is a solution for $I$.

After step 0, Invariances~1,2 are trivially satisfied. At step $j$,
let $W[i]$ be the earliest window in $W$ such that $S[i]\ne S'[i]$.
Then $S[i]>S'[i]=S[k]$ for some $k\in[i+1..n]$ because $S[1..i-1]=S'[1..i-1]$
and $S'$ is ordered. Also, $i>j-1$ because of Invariance~1. Swap
$S[i]$ and $S[k]$. $S$ is still a solution for $I$, because before
the swap $W[i]\le W[k]$ and $S[i]>S[k]$. After the swap, $S[i]=S'[i]$
and hence $S[1..i]=S'[1..i]$. Therefore Invariances~1,2 are preserved.

Now we shall prove (2). Take any ordered instance $I=(n,T,W)$ and
valid allocation $S$ for $I$ and ordered sequence $k[1..m]$. Let
$X$ be $S$ with $S[k[1..m]]$ sorted, and let $j[1..a]$ be the
subsequence of $k$ such that $T[j[1..a]]$ are exactly the allocated
tasks in $S[k[1..m]]$. Then $S[j[1..a]]$ is a solution for $I'=(a,T[j[1..a]],W[j[1..a]])$
and $X[j[1..a]]$ is $S[j[1..a]]$ sorted. Thus $X[j[1..a]]$ is an
ordered solution for $I'$ by (1), and hence $X[k[1..m]]$ is an ordered
valid allocation for $(m,T[k[1..m]],W[k[1..m]])$. Since $X$ does
not have any overlapping slots, $X$ is a valid allocation for $I$.\end{proof}
\begin{rem*}
\nameref{F.Ord} and its proof applies with no change to the $p$-processor
case, because sorting does not affect the processor and position for
each slot. Additionally, it is easy to see that any ordered solution
can be made into a cyclic one, namely that it has exactly the same
allocated slots but the allocated processor cycles with the window
rank modulo $p$.\end{rem*}
\begin{proc}[Leftmost]
\label{F.Left}~
\begin{block}
\item \textbf{\uline{Implementation}}

\begin{block}
\item \textbf{Procedure Leftmost( instance $I=(n,T,W)$ ):}

\begin{block}
\item Set $S[0]\endd=-\infty$.
\item For $i$ from $1$ up to $n$:

\begin{block}
\item Set $S[i]\start=\max(W[i]\start,S[i-1]\endd)$.
\item Set $S[i]\endd=S[i]\start+1$.
\end{block}
\item Return $S[1..n]$.
\end{block}
\end{block}
\end{block}
\end{proc}
\begin{thm}[Leftmost's properties]
\label{F.Left.prop} Take any feasible ordered instance $I$. Let
$S=\f{Leftmost}(I)$. Then the following hold:
\begin{enumerate}
\item $S$ is an ordered solution for $I$.
\item For any ordered solution $X$ for $I$, we have $S[i]\le X[i]$ for
any $i\in[1..n]$.
\end{enumerate}
\end{thm}
\begin{proof}
Consider any ordered solution $X$ for $I$, and set $S[0]=X[0]=(-\infty,-\infty)$.
We shall inductively prove that $S[i]\le X[i]$ for each $i\in[1..n]$.
For each $i$ from $1$ to $n$, either of the following cases hold:
\begin{itemize}
\item $W[i]\start\ge S[i-1]\endd$:

\begin{block}
\item Then $S[i]\start=W[i]\start\le X[i]\start$, because $X$ is a solution.
\end{block}
\item $W[i]\start<S[i-1]\endd$:

\begin{block}
\item Then $S[i]\start=S[i-1]\endd\le X[i-1]\endd\le X[i]\start$, because
$S[i-1]\le X[i-1]$ by induction and $X$ is an ordered solution.
\end{block}
\end{itemize}
In both cases $S[i]\le X[i]$. Therefore by induction (2) follows.
Also, for any $i\in[1..n]$, we have both $S[i]\start\ge W[i]\start$
and $S[i]\endd\le X[i]\endd\le W[i]\endd$, and hence $S[i]\wi W[i]$.
Finally by construction $S[i]\start\ge S[i-1]\endd$ for any $i\in[1..n]$,
hence $S$ is an ordered solution. Since such an $X$ exists by \nameref{F.Ord}
(\ref{F.Ord}), (1) follows.\end{proof}
\begin{rem*}
\nameref{F.Left} has an analogous version for multiple processors,
which assigns for each task in order the leftmost possible (processor,slot)
pair that does not overlap previous assignments, and has exactly the
same properties, with an analogously modified proof.\end{rem*}
\begin{proc}[Near]
\label{F.Near}~
\begin{block}
\item \textbf{\uline{Dependencies}}

\begin{block}
\item \textbf{\nameref{F.Left} (\ref{F.Left})}
\end{block}
\item \textbf{\uline{Implementation}}

\begin{block}
\item \textbf{Procedure Near( ordered insert state $(I=(n,T,W),S,r)$ with
ordered $S$ ):}

\begin{block}
\item Set $L=\f{Leftmost}(I)$.
\item Return $\cond{L[r]\ge S[r-1]}L{L[1..r-1]\cdot\f{seq}(S[r-1])\cdot S[r+1..n]}$.
\end{block}
\end{block}
\end{block}
\end{proc}
\begin{defn}[Near ordered solution]
 Take any feasible ordered insert state $(I=(n,T,W),S,r)$ with ordered
$S$. Call $N$ a \textbf{near ordered solution for $(I,S,r)$} iff
all the following hold:
\begin{itemize}
\item $N$ is an ordered solution for $I$.
\item $S[r-1]\le N[r]\le S[r+1]$.
\end{itemize}
\end{defn}
\begin{thm}[Near's properties]
\label{F.Near.prop} Take any feasible ordered insert state $(I=(n,T,W),S,r)$
with ordered $S$. Let $N=\f{Near}(I,S,r)$. Then $N$ is a near ordered
solution for $(I,S,r)$.\end{thm}
\begin{proof}
Let $L=\f{Leftmost}(I)$. Either of the following cases hold:
\begin{itemize}
\item $L[r]\ge S[r-1]$:

\begin{block}
\item By \nameref{F.Left.prop} (\ref{F.Left.prop}), $L$ is an ordered
solution for $I$ and $L[r]\start=\max(W[r]\start,L[r-1]\endd)$ $\le\max(W[r+1]\start,S[r-1]\endd)\le S[r+1]\start$,
and hence $S[r-1]\le L[r]\le S[r+1]$. Also, $N=L$.
\end{block}
\item $L[r]<S[r-1]$:

\begin{block}
\item Then $L[r-1]\endd\le L[r]\start<S[r-1]\start$, and hence $N=L[1..r-1]\cdot\f{seq}(S[r-1])\cdot S[r+1..n]$
is an ordered solution for $I$ and $S[r-1]=N[r]<S[r+1]$.
\end{block}
\end{itemize}
Therefore in both cases $N$ has the properties claimed.\end{proof}
\begin{rem*}
The multi-processor version of \nameref{F.Near} combines a \nameref{F.Left}
solution and a Rightmost solution (defined symmetrically to Leftmost)
that agree on the processor for the inserted task to obtain a solution
with the desired properties for similar reasons.\end{rem*}
\begin{proc}[Snap]
\label{F.Snap}~
\begin{block}
\item \textbf{\uline{Implementation}}

\begin{block}
\item \textbf{Procedure Snap( slot $s$ , window $w$ ):}

\begin{block}
\item If $s\start<w\start$:

\begin{block}
\item Return $w\start+[0,1]$.
\end{block}
\item If $s\endd>w\endd$:

\begin{block}
\item Return $w\endd+[-1,0]$.
\end{block}
\item Return $s$.
\end{block}
\end{block}
\end{block}
\end{proc}
\begin{thm}[Snap's properties]
\label{F.Snap.prop} Take any slot $s$ and window $w$. Let $r=\f{Snap}(s,w)$.
Then the following properties hold:
\begin{enumerate}
\item $r\wi w$.
\item $r\start\le\max(w\start,s\start)$.
\item $r\endd\ge\min(w\endd,s\endd)$.
\end{enumerate}
\end{thm}
\begin{proof}
All properties are obvious by construction.
\end{proof}

\subsection{Lower bound}

RA is asymptotically optimal for an allocator that keeps the solution
in order, but we can do much better if the solution does not have
to be kept in order. But before we describe such an allocator, we
will first present for any $\e\in(0,1)$ an $\e$-slack feasible insert
state on which any allocator will take at least $\floor{\log_{1+\e}(\frac{1}{\e})}$
reallocations, which is $\Q\left(\frac{1}{\e}\log(\frac{1}{\e})\right)$
as $\e\to0$. There are worse situations that need approximately double
that, but the insert state given here is much easier to analyze, giving
a lower bound on any allocator that can work on arbitrary feasible
insert states. In this insert state there are $\floor{\frac{c-1}{\e}}$
windows in order and the inserted window is the earliest. Each window
has the smallest non-negative start position possible, as constrained
by the existence of a $(1+\e)$-partial solution. The partial solution
is also ordered with the $k$-th allocated slot at $[k-1,k]$. The
details are given below.
\begin{thm}[Lower Bound]
\label{F.LowerBound} Take any $\e\in(0,1)$ and any $c\ge\frac{1}{\e}+1$.
Then there is some feasible $\e$-slack insert state with window length
$c$ such that any allocator $A$ that solves it makes at least $\floor{\log_{1+\e}(\frac{1}{\e})}$
reallocations.\end{thm}
\begin{proof}
(In this proof we shall omit the derivation of purely algebraic inequalities
involving $\e$,$c$,$n$ as they can be easily verified.) First let
$I=(n,T,W)$ where $n=\floor{\frac{c-1}{\e}}$ and $W[i]=\max(c,i(1+\e))+[-c,0]$
for each $i\in[1..n]$. Then $I$ is an $\e$-slack instance, because
it has a $(1+\e)$-solution $E$ where $E[i]=[i-1,i](1+\e)$ for each
$i\in[1..n]$. This is easy to check as follows. Firstly, $E[1..n]$
are non-overlapping. Secondly, for any $i\in[1..n]$, $E[i]\wi W[i]$
because:
\begin{itemize}
\item $E[i]\start=(i-1)(1+\e)\ge\max(0,i(1+\e)-c)=W[i]\start$.
\item $E[i]\endd=i(1+\e)\le W[i]\endd$.
\end{itemize}
Now let $S[i]=[i-1,i]$ for each $i\in[1..n]$. Then $S$ is a solution
for $I$, since $S[1..n]$ are non-overlapping, and for any $i\in[1..n]$,
$S[i]\wi W[i]$ because:
\begin{itemize}
\item $S[i]\start=i-1\ge\max(0,i(1+\e)-c)=W[i]\start$.
\item $S[i]\endd=i\le i(1+\e)\le W[i]\endd$.
\end{itemize}
After an insertion into $I$ of a new task $t$ with window $[0,c]$,
the resulting insert state is feasible, since $t$ can be allocated
to $[0,1]\wi[0,c]$, and for each $i\in[1..n]$, $T[i]$ can be allocated
to $[i,i+1]$ because:
\begin{itemize}
\item $i>S[i]\start\ge W[i]\start$.
\item $i+1\le\max(c,i(1+\e))=W[i]\endd$.
\end{itemize}
Now consider any allocator $A$ that solves such an insert state.
Let $S'[0]$ be the slot that $A$ will allocate $t$ to, and $S'[1..n]$
be the slots that $A$ will allocate $T[1..n]$ to respectively. Let
$l=\floor{\log_{1+\e}(\frac{1}{\e})}$. Set $k[0]=0$. We shall construct
$k[1..l]$ iteratively such that the following invariances hold after
each step $j$ from $0$ to $l$:
\begin{enumerate}
\item $A$ will reallocate $T[k[1..j]]$.
\item $S'[k[0..j]]\endd$ is strictly increasing.
\item $S'[k[j]]\endd\le c(1+\e)^{j}$.
\end{enumerate}
After step $0$, Invariances~1,2 are trivially satisfied, and Invariance~3
is satisfied because $S'[k[0]]\endd=S'[0]\endd\le c\le c(1+\e)^{0}$.
To construct $k[1..l]$, at step $j$ from $1$ to $l$ construct
$k[j]$ as follows. Set $m=\floor{S'[k[j-1]]\endd}$. Then, by Invariance~3,
$S'[k[j-1]]\endd\le c(1+\e)^{j-1}\le c(1+\e)^{l-1}\le c(1+\e)^{\log_{1+\e}(\frac{1}{\e})-1}=\frac{c}{\e(1+\e)}\le\frac{c-1}{\e}$
and hence $m\le\floor{\frac{c-1}{\e}}=n$. Thus $\class{S'[i]}{i\in[0..m]}$
are slots that will be allocated by $A$, and hence $\max_{i\in[1..m]}S'[i]\endd\ge m+1$
because:
\begin{itemize}
\item $\max_{i\in[0..m]}S'[i]\endd\ge m+1$, since $S'[0..m]$ do not overlap
and start no earlier than $0$.
\item $S'[0]\endd\le S'[k[j-1]]\endd<m+1$, since $S'[0..j-1]$ is strictly
increasing by Invariance~2.
\end{itemize}
Now set $k[j]\in[1..m]$ such that $S'[k[j]]\endd\ge m+1$. We then
verify that Invariances~1,2,3 are preserved. Firstly, $A$ will reallocate
$T[k[j]]$, because $S'[k[j]]\endd>m\ge k[j]=S[k[j]]\endd$. Secondly,
$S'[k[j]]\endd\ge m+1=\floor{S'[k[j-1]]\endd}+1>S'[k[j-1]]\endd$.
Thirdly, $S'[k[j]]\endd\le W[k[j]]\endd=\max(c,k[j](1+\e))\le c(1+\e)^{j}$,
because $k[j]\le m\le S'[k[j-1]]\endd\le c(1+\e)^{j-1}$ by Invariance~3.

Therefore $A$ will reallocate $T[k[1..l]]$ by Invariance~1, which
are distinct because $S'[k[1..l]]$ are distinct by Invariance~2,
and hence $A$ will reallocate at least $l=\floor{\log_{1+\e}(\frac{1}{\e})}$
tasks.\end{proof}
\begin{rem*}
\nameref{F.LowerBound} (\ref{F.LowerBound}) is tight for allocators
that are required to work on any insert state, as we shall show in
the subsequent section, but if an allocator is used on the system
from the beginning, we do not know if it is possible to do better
by avoiding such `bad' insert states.
\end{rem*}

\subsection{Optimal single-processor fixed-window allocator}

We will now describe a single-processor allocator FA that can solve
any $\e$-slack feasible insert state using at most $\max\left(2\log_{1+\frac{1}{2}\e}\left(\frac{14}{\e^{2}}\right)+\frac{34}{\e}+6,14\right)$
reallocations, which is $O\left(\frac{1}{\e}\log(\frac{1}{\e})\right)$
as $\e\to0$ and hence FA is asymptotically optimal for general insert
states. If FA is allowed to maintain an internal state, it would take
only $O\left(\frac{1}{\e}\log(\frac{1}{\e})\log(n)+\log(n)^{2}\right)$
time on each feasible insertion, and will return failure in $O(\log(n))$
time on an infeasible insertion.

We first give an outline of the main ideas behind FA's insertion procedure:
\begin{enumerate}
\item If the surrounding region (around the window of the task to be allocated)
has enough empty space, we \textbf{pack} the slots in that region
to create a gap into which we can squeeze one more slot.
\item If the surrounding region has not enough empty space, we need to first
`get' to some region with sufficient empty space using two mechanisms:

\begin{enumerate}
\item \textbf{Jumping}: A jump allocates a task to a slot that had already
been allocated to another task whose window is as far as possible
in some direction, displacing the other task, which then has to be
reallocated, potentially in the next jump.
\item \textbf{Pushing}: When jumps are not possible, we have to use pushes
first, where pushing a slot in some direction is simply to shift it
just barely enough to make space for the inserted task or the previous
pushed slot. The pushed slot may overlap the slot of another task,
which may then be reallocated in a subsequent push or jump.
\end{enumerate}
\item In both the pushing and packing phases, in order to guarantee that
they do not reallocate tasks outside their windows, it is very important
to have the slots involved be in the same order as their windows,
which we achieve by swaps.
\item We do not actually do the jumping, but just simulate it to see where
we can `get' to. Only after we have finished the final packing phase
do we perform actual jumps from the inserted task to the gap created.
\item If the inserted task is allocated correctly, at most $O(\frac{1}{\e})$
slots need to be pushed before jumping can be carried out. So in some
cases it is necessary to make a second attempt to find such a slot
for the inserted task if the first attempt takes too many reallocations.
\item We allow the region for the packing phase to contain up to $\frac{2}{\e}$
slots each, for two reasons:

\begin{enumerate}
\item It means that we will pack at most $O(\frac{1}{\e})$ slots.
\item When packing is impossible, the span of all the windows reached so
far by the simulated jumping will be less than about $(1+\frac{1}{2}\e)$
times the number of slots within it. Thus the existence of a $(1+\e)$-solution
makes the span grow by a factor of about $(1+\frac{1}{2}\e)$ on each
jump, which is the fastest possible in the worst case to within a
constant factor.
\end{enumerate}
\end{enumerate}
These ideas may sound simple, but it is extremely tricky to actually
make them work. So now we shall give a step-by-step high-level description
of FA's insertion procedure.

FA keeps the windows $W$ in order in an IDSetRQ, so that it can query
for any range of consecutive windows their maximum $(\text{start point}-\text{rank within the range})$
and their minimum $(\text{end point}-\text{rank within the range})$.
For any inserted task with window $w$, FA can easily determine the
rank $r$ of $w$ if it is inserted into $W$, and can then by the
following obtain in $O(\log(n))$ time the interval $[istart,iend]$
such that for any unit interval $s$, there is an ordered solution
that allocates the new task to $s$ if and only if $s$ is contained
within $[istart,iend]$:
\begin{itemize}
\item $istart=(\max(\text{start point}-\text{rank})\text{ over windows of rank}\le r)+r$
\item $iend=(\min(\text{end point}-\text{rank})\text{ over windows of rank}\ge r)+r$
\end{itemize}
This means that FA can easily check whether the insertion is feasible,
because by the Ordering theorem any instance is feasible if and only
if it has an ordered solution, which is equivalent to $iend-istart\ge1$.

FA also keeps $(\text{slot},\text{window})$ pairs in an IDSetRQ in
order of the slots, so that it can query for any range of consecutive
slots their tasks' earliest window and latest window. This is a key
ingredient for an efficient implementation of the jumping part of
the algorithm, where each \textbf{jump} allocates a task to a slot
that had already been allocated to another task whose window is as
far as possible in some direction, displacing the other task, which
then has to be reallocated. A sequence of jumps will cause a cascade
of reallocations that propagate as fast as possible, so that space
usage can be efficiently reorganized.

If FA knew what $\e$ was, it could just run the appropriate subroutine:
\begin{itemize}
\item $c\ge\frac{7}{\e}+4$: \textbf{LargeWindow}
\item $c<\frac{7}{\e}+4$: \textbf{SmallWindow}
\end{itemize}
LargeWindow is named thus because the window length is large enough
to ensure that jumping can begin immediately. It simulates jumps from
the inserted task's window to the furthest windows in both directions
simultaneously on each jump. It stops when it finds $\frac{2}{\e}$
or fewer consecutive allocated slots in an interval within the current
span with a total empty space of at least $1$, where the current
\textbf{span} is defined as the span of all windows reached by the
jumps. Then it sorts those slots and packs them greedily to create
a unit gap in the middle. Finally it performs actual jumps from the
inserted window to the unit gap, which solves the instance.

To find such a set of slots, FA divides the intervening gaps (touching
slots are considered to have a gap of length $0$ in-between) into
blocks of $(\frac{2}{\e}+1)$ consecutive gaps each, except the last
block, and uses the IDSetRQ containing the $(\text{slot},\text{window})$
pairs to determine in $O(\log(n))$ time if the blocks have average
total space at least $1$. If so, FA uses a binary search to find
a block with at least average total space in $O(\log(n)^{2})$ time
using the same IDSetRQ. Such a block will have total space at least
$1$.

SmallWindow on the other hand may need to use a number of pushes before
jumping is even possible, where a push reallocates a task by shifting
its slot just enough to fix an overlap. SmallWindow may need to make
two attempts. Letting $X$ be the partial solution sorted, it separately
tries two positions for the inserted slot $s$:
\begin{itemize}
\item Within $[istart,iend]$ and nearest to $X[r-1]+1$
\item Within $[istart,iend]$ and nearest to $X[r+1]-1$
\end{itemize}
One of them is guaranteed to succeed due to the restrictions that
any $(1+\e)$-solution for the previous instance place on the subsequent
steps, but it is not clear how it can be determined efficiently without
even knowing $\e$. Thus SmallWindow simply tries both. In each case,
SmallWindow sorts the allocated slots that are within the inserted
window, and now $s$ `separates' the allocated slots, namely that
there is no slot that is after $s$ but with earlier window, or before
$s$ but with later window, because $X[r-1]\le s\le X[r+1]$. SmallWindow
then pushes $m$ neighbouring slots aside on each side as necessary,
for each of them after swapping it into the window with the same rank,
so that each of the pushed slots also `separate' the allocated slots.

On each side separately, if the last pushed slot still overlaps the
next one, the next one is deallocated and jump simulation is begun
in the pushing direction with its window as the initial window. For
the right side, the current span is defined to start at the end of
the last pushed slot, and for the left side defined to end at the
start of the last pushed slot. Everything else follows LargeWindow
exactly.

Both LargeWindow and SmallWindow make only $O\left(\frac{1}{\e}\log(\frac{1}{\e})\right)$
jumps due to the following reasons. Firstly, on each jump the current
span grows by a factor of at least about $1+\frac{\e}{2}$ when it
is sufficiently large, because the average space per block of $(\frac{2}{\e}+1)$
gaps within it must be less than $1$, making the slots within the
current span cramped. Secondly, the existence of a $(1+\e)$-solution
for the previous instance forces the current span to grow to roughly
proportional to that number of slots. For LargeWindow, the starting
span is already large enough. For SmallWindow, it is enough that the
starting span just be at least $2+\e$, which is ensured by the pushing
phase. In either procedure, the current span cannot increase by more
than $c-1$ on each jump, and hence the number of jumps ends up being
$O\left(\log_{1+\frac{1}{2}\e}(\frac{1}{\e})\right)$ for LargeWindow
and $O\left(\log_{1+\frac{1}{2}\e}(\frac{c}{\e})\right)$ for SmallWindow.

The problem is that it is probably impossible to determine $\e$ exactly
in $O\left(\frac{1}{\e}\log(\frac{1}{\e})\log(n)\right)$ time, so
FA uses the following standard trick to avoid having to know $\e$
at all. On a feasible insertion, FA starts by setting $e=2$, and
assumes that $\e=e$ in order to run the appropriate procedure as
before, but limiting the number of jumps to the maximum it should
be. If it fails, it must be that $\e<e$, so FA halves $e$ and tries
again, repeating until $\frac{2}{e}\ge n$, at which point both LargeWindow
and SmallWindow would definitely succeed. Each failed trial takes
$O\left(\frac{1}{e}\log(\frac{1}{e})\log(n)\right)$ time and hence
all the failed trials take $O\left(\frac{1}{\e}\log(\frac{1}{\e})\log(n)\right)$
time by a simple summation. The one successful trial takes $O\left(\frac{1}{\e}\log(\frac{1}{\e})\log(n)+\log(n)^{2}\right)$
time.
\begin{algo}[FA]
\label{F.FA}~
\begin{block}
\item \textbf{\uline{Dependencies}}

\begin{block}
\item \textbf{\nameref{F.FA.LW} (\ref{F.FA.LW})}
\item \textbf{\nameref{F.FA.SW} (\ref{F.FA.SW})}
\end{block}
\item \textbf{\uline{Variables}}

\begin{block}
\item Ordered instance $I=(n,T,W)$ // current ordered instance ; must be
feasible before and after each operation
\item Allocation $S$ // current allocation for $I$ ; must be a solution
for $I$ before and after each operation
\end{block}
\item \textbf{\uline{Initialization}}

\begin{block}
\item Set $I=(n,T,W)=(0,(),())$ and $S=()$.
\item Initialize an IDSetRQ for $W$, with the range query function $(\text{maximum}(\text{start}-\text{rank}),\text{minimum}(\text{end}-\text{rank}))$.
\item Initialize an IDSetRQ for $(S,W)$ sorted by $S$, with range query
function $(\text{earliest window},\text{latest window})$.
\end{block}
\item \textbf{\uline{External Interface}}

\begin{block}
\item \textbf{Procedure Insert( task $t$ , window $w$ )} // inserts task
$t$ with window $w$ into the system
\item \textbf{Procedure Delete( task $t$ )} // deletes task $t$ from the
system
\end{block}
\item \textbf{\uline{Implementation}}

\begin{block}
\item \textbf{Procedure Insert( task $t$ , window $w$ ):}

\begin{block}
\item // Create the insert state //
\item Backup $I$,$S$.
\item Set $(I=(n,T,W),S,r)$ to be the ordered insert state on insertion
of $(t,w)$ into $(I,S)$.
\item // Find the range of possible insertion points in an ordered solution
for $I$ //
\item Set $istart=\max_{i\in[1..r]}(W[i]\start-i)+r$.
\item Set $iend=\min_{i\in[r..n]}(W[i]\endd-(i-r))$.
\item // Check if the insertion is feasible //
\item If $iend-istart<1$:

\begin{block}
\item Restore $I$,$S$.
\item Return $Failure$.
\end{block}
\item // Perform doubling on $m$ from 1 up //
\item Backup $S$.
\item For $m$ doubling from $1$ up to $2n$:

\begin{block}
\item Restore $S$.
\item // Set $e=\frac{2}{m}$ and assume $\varepsilon=e$ and use the appropriate
procedure based on $c$ and $e$ //
\item Set $e=\frac{2}{m}$.
\item If $c\ge\frac{7}{e}+4$:

\begin{block}
\item If $\f{LargeWindow}((I,S,r),m)=Success$:

\begin{block}
\item Return $Success$.
\end{block}
\end{block}
\item Otherwise:

\begin{block}
\item If $\f{SmallWindow}((I,S,r),m,istart,iend)=Success$:

\begin{block}
\item Return $Success$.
\end{block}
\end{block}
\end{block}
\item // This will never be reached //
\end{block}
\item \textbf{Procedure Delete( task $t$ ):}

\begin{block}
\item If $t\in T$:

\begin{block}
\item Delete $t$ from $(I,S)$.
\item Return $Success$.
\end{block}
\item Otherwise:

\begin{block}
\item Return $Failure$.
\end{block}
\end{block}
\end{block}
\end{block}
\end{algo}
\begin{thm}[FA's properties]
\label{F.FA.prop} On an insertion of a new task, if the current
instance is $\e$-slack, FA has the following properties:
\begin{itemize}
\item If the insertion is feasible, it returns $Success$ after updating
$S$ to be a solution for the new instance by making at most $\max\left(2\log_{1+\frac{1}{2}\e}\left(\frac{14}{\e^{2}}\right)+\frac{34}{\e}+6,14\right)$
reallocations and taking $O\left(\frac{1}{\e}\log(\frac{1}{\e})\log(n)+\log(n)^{2}\right)$
time.
\item If the insertion is not feasible, it returns $Failure$ in $O(\log(n))$
time.
\end{itemize}
\end{thm}
\begin{proof}
By \nameref{F.FA.IR} (\ref{F.FA.IR}), on such an insertion $[istart,iend]$
will be set such that any slot is within that range iff the new task
can be allocated to it in some ordered solution. This implies that
$iend-istart\ge1$ iff the insertion is feasible. Thus FA will restore
the previous state and return $Failure$ iff the insertion is not
feasible. Also, this feasibility check takes $O(\log(n))$ time as
it only needs one range associative query to obtain $istart,iend$.
If the insertion is feasible, FA will enter the for-loop, which by
iteratively doubling $m$ and halving $e$ will ensure that a reasonably
good solution will be found and yet only the last iteration takes
a significant portion of the total time. The reason is that both \nameref{F.FA.LW}
(\ref{F.FA.LW}) and \nameref{F.FA.SW} (\ref{F.FA.SW}) have the
following properties (\ref{F.FA.LW.prop},\ref{F.FA.SW.prop}) when
called here under the respective conditions $c\ge\frac{7}{e}+4$ and
$c<\frac{7}{e}+4$:
\begin{itemize}
\item If $m\ge n$ or $e\le\e$, it will solve the insert state using at
most $\left(2\floor{\max\left(\log_{1+\frac{1}{2}\e}\left(\frac{14}{\e^{2}}\right),0\right)}+\frac{17}{e}+6\right)$
reallocations and return $Success$ in $O\left(\frac{1}{e}\log(\frac{1}{e})\log(n)+\log(n)^{2}\right)$
time.
\item If it returns $Failure$, it would have taken $O\left(\frac{1}{e}\log(\frac{1}{e})\log(n)\right)$
time.
\end{itemize}
Since the while-loop will run at least once with $m\in[n,2n)$ at
the start of the loop, FA will always return $Success$ in some run
of the loop and never exit the loop otherwise. Since $\frac{1}{e}\log(\frac{1}{e})$
more than doubles when $e$ is halved, and $e=2$ or $e>\frac{1}{2}\e$
in the successful loop, we obtain $2\floor{\max\left(\log_{1+\frac{1}{2}\e}\left(\frac{14}{\e^{2}}\right),0\right)}+\floor{\frac{17}{e}}+6\le\max\left(2\log_{1+\frac{1}{2}\e}\left(\frac{14}{\e^{2}}\right)+\frac{34}{\e}+6,14\right)$
and the total time taken is $O\left(\frac{1}{\e}\log(\frac{1}{\e})\log(n)+\log(n)^{2}\right)$
as $\e\to0$.\end{proof}
\begin{lem}[Insertion Range]
\label{F.FA.IR} Take any ordered insert state $(I=(n,T,W),S,r)$,
and define $istart,iend$ as follows:
\begin{itemize}
\item $istart=\max_{i\in[1..r]}(W[i]\start-i+r)$.
\item $iend=\min_{i\in[r..n]}(W[i]\endd-i+r)$.
\end{itemize}
Take any slot $s$. Then $s\wi[istart,iend]$ iff $s=Y[r]$ for some
ordered solution $Y$ for $I$.\end{lem}
\begin{proof}
Firstly if $s\wi[istart,iend]$, then let $X$ be $S$ sorted, which
by \nameref{F.Ord} (\ref{F.Ord}) is a partial solution for $(I,r)$,
and let $Y$ be as follows:
\begin{itemize}
\item $Y[r]=s$.
\item $Y[i]=\min(X[i],s-r+i)$ for each $i\in[1..r-1]$.
\item $Y[i]=\max(X[i],s-r+i)$ for each $i\in[r+1..n]$.
\end{itemize}
Then $Y$ is an ordered solution for $I$, which we can check as follows:
\begin{itemize}
\item $Y$ allocates each task to a slot within its window.

\begin{itemize}
\item $Y[r]\wi[istart,iend]\wi W[r]$.
\item $Y[i]\wi[\min(X[i]\start,istart-r+i),X[i]\endd]\wi W[i]$ for each
$i\in[1..r-1]$.
\item $Y[i]\wi[X[i]\start,\max(X[i]\endd,iend-r+i)]\wi W[i]$ for each $i\in[r+1..n]$.
\end{itemize}
\item $Y$ allocates different tasks to non-overlapping slots.

\begin{itemize}
\item $Y[i]\endd=\min(X[i]\endd,s\endd-r+i)\le\min(X[i+1]\start,s\start-r+(i+1))=Y[i]\start$
for each $i\in[1..r-2]$.
\item $Y[i]\start=\max(X[i]\start,s\start-r+i)\ge\max(X[i-1]\endd,s\endd-r+(i-1))=Y[i]\endd$
for each $i\in[r+2..n]$.
\item $Y[r-1]\endd\le s\endd-1=Y[r]\start$.
\item $Y[r+1]\start\ge s\start+1=Y[r]\endd$.
\end{itemize}
\end{itemize}
Conversely if $s=Z[r]$ for some ordered solution $Z$ for $I$, then
$s\wi[istart,iend]$, which we can check as follows:
\begin{itemize}
\item $s\start=Z[r]\start=\max_{i\in[1..r]}(Z[i]\start-r+i)\ge\max_{i\in[1..r]}(W[i]\start-r+i)=istart$.
\item $s\endd=Z[r]\endd=\min_{i\in[r..n]}(Z[i]\endd-r+i)\le\min_{i\in[r..n]}(W[i]\endd-r+i)=iend$.
\end{itemize}
Therefore the desired equivalence follows.\end{proof}
\begin{rem*}
This technique to find the possible range for $Y[r]$ in an ordered
solution $Y$ can be extended to the $p$-processor case to take $O(p\log(n))$
time per operation with judicious use of data structures, since for
any ordered solution the cyclic one has exactly the same slots, and
the range query in the IDSetRQ can be suitably modified to compute
the values for each set of windows with the same rank modulo $p$.\end{rem*}
\begin{subr}[LargeWindow]
\label{F.FA.LW}~
\begin{block}
\item \textbf{\uline{Dependencies}}

\begin{block}
\item \textbf{\nameref{F.FA.Jump} (\ref{F.FA.Jump})}
\end{block}
\item \textbf{\uline{Implementation}}

\begin{block}
\item \textbf{Subroutine LargeWindow( insert state $(I=(n,T,W),S,r)$ ,
nat $m=\frac{2}{e}$ ):}

\begin{block}
\item // Jump in both directions //
\item Set $jumps=\max\left(\log_{1+\frac{1}{2}e}\left(\frac{28}{e}\right),0\right)+1$.
\item Return $\f{Jump}((I,S,r),m,W[r],\str{both},jumps)$.
\end{block}
\end{block}
\end{block}
\end{subr}
\begin{lem}[LargeWindow's properties]
\label{F.FA.LW.prop} Take any $\e$-slack ordered insert state $(I,S,r)$
and positive natural $m=\frac{2}{e}$. Then $\f{LargeWindow}((I,S,r),m)$
does the following:
\begin{itemize}
\item If $c\ge\frac{7}{e}+4$ and ( $m\ge n$ or $e\le\e$ ), it solves
$(I,S,r)$ using at most $\max\left(\log_{1+\frac{1}{2}\e}\left(\frac{28}{\e}\right),0\right)+\frac{2}{e}+1$
reallocations and returns $Success$ in $O\left(\frac{1}{e}\log(\frac{1}{e})\log(n)+\log(n)^{2}\right)$
time.
\item If it returns $Failure$, it would have taken $O\left(\frac{1}{e}\log(\frac{1}{e})\log(n)\right)$
time.
\end{itemize}
\end{lem}
\begin{proof}
The lemma follows directly from \nameref{F.FA.Jump.prop} (\ref{F.FA.Jump.prop}~Properties~1,2),
where $d=7$. It only needs to be checked that $c\ge\frac{6}{\e}+3+\frac{c-1}{7}$,
which follows from $\frac{6c+1}{7}>\frac{6}{7}(\frac{7}{e}+4)>\frac{6}{e}+3$.\end{proof}
\begin{subr}[SmallWindow]
\label{F.FA.SW}~
\begin{block}
\item \textbf{\uline{Dependencies}}

\begin{block}
\item \textbf{\nameref{F.FA.Push} (\ref{F.FA.Push})}
\end{block}
\item \textbf{\uline{Implementation}}

\begin{block}
\item \textbf{Subroutine SmallWindow( insert state $(I=(n,T,W),S,r)$ ,
nat $m=\frac{2}{e}$ , real $istart$ , real $iend$ ):}

\begin{block}
\item Define $X$ to be $S$ sorted with $(X[0]\endd,X[n+1]\start)=(-\infty,\infty)$.
\item Set $a=\f{Snap}(X[r+1]-1,[istart,iend])$.
\item Set $b=\f{Snap}(X[r-1]+1,[istart,iend])$.
\item Backup $I,S$.
\item If $\f{Push}((I,S,r),m,a)=Success$:

\begin{block}
\item Return $Success$.
\end{block}
\item Restore $I,S$.
\item If $\f{Push}((I,S,r),m,b)=Success$:

\begin{block}
\item Return $Success$.
\end{block}
\item Restore $I,S$.
\item Return $Failure$.
\end{block}
\end{block}
\end{block}
\end{subr}
\clearpage{}
\begin{lem}[SmallWindow's properties]
\label{F.FA.SW.prop} Take any $\e$-slack ordered insert state $(I,S,r)$
and positive natural $m=\frac{2}{e}$ and reals $istart,iend$. Then
$\f{SmallWindow}((I,S,r),m,istart,iend)$ does the following:
\begin{itemize}
\item If $c<\frac{7}{e}+4$ and ( $m\ge n$ or $e\le\e$ ) and ( $s\wi[istart,iend]$
iff $s=Y[r]$ for some ordered solution $Y$ for $I$ ) for any slot
$s$, it solves $(I,S,r)$ using at most $2\floor{\max\left(\log_{1+\frac{1}{2}\e}\left(\frac{14}{\e^{2}}\right),0\right)}+\frac{17}{e}+6$
reallocations and returns $Success$ in $O\left(\frac{1}{e}\log(\frac{1}{e})\log(n)+\log(n)^{2}\right)$
time.
\item If it returns $Failure$, it would have taken $O\left(\frac{1}{e}\log(\frac{1}{e})\log(n)\right)$
time.
\end{itemize}
\end{lem}
\begin{proof}
The short summary is that $\f{Push}((I,S,r),m,s)$ succeeds with the
right choice of $s$ if $I$ is feasible, and it turns out that it
is enough to try just $a$ and $b$ specified in the subroutine.

Let $E$ be some ordered $(1+\e)$-partial solution for $(I,r)$.
Then $istart\le E[r+1]\start$, otherwise there is some ordered solution
$Y$ for $I$ such that $Y[r]\start>E[r+1]\start$ and so concatenating
$\seq{E[i]\start+[0,1]}{i\in[1..r-1]}$, $(E[r+1]\start+[0,1])$ and
$Y[r+1..n]$ gives an ordered solution for $I$ that implies $istart\le E[r+1]\start$.
Likewise $iend\ge E[r-1]\endd$.

Next let $N$ be some near ordered solution for $(I,S,r)$ by \nameref{F.Near.prop}
(\ref{F.Near.prop}). Then $istart\le N[r]\start\le X[r+1]\start$
and $iend\ge N[r]\endd\ge X[r-1]\endd$. From these we get $X[r-1]\le a,b\le X[r+1]$
because of the following inequalities arising from \nameref{F.Snap.prop}
(\ref{F.Snap.prop}):
\begin{itemize}
\item $a\le\max(X[r+1]-1,istart+[0,1])\le X[r+1]$.
\item $a\ge\min(X[r+1]-1,iend+[-1,0])\ge X[r-1]$.
\item $b\le\max(X[r-1]+1,istart+[0,1])\le X[r+1]$.
\item $b\ge\min(X[r-1]+1,iend+[-1,0])\ge X[r-1]$.
\end{itemize}
Also, $b\endd-a\start\le\max(X[r-1]\endd,istart)-\min(X[r+1]\start,iend)+2\le2$,
which gives $(E[r+1]\endd-b\endd)+(a\start-E[r-1]\start)$ $\ge2(1+\e)-2=2\e$
and hence either $E[r+1]\endd-b\endd\ge\e$ or $a\start-E[r-1]\start\ge\e$.
In addition, if $a$ overlaps $X[r+1]$, then $a\start=istart\le E[r+1]\start$
and so $E[r+1]\endd-a\endd\ge\e$. Likewise, if $b$ overlaps $X[r-1]$,
then $b\start-E[r-1]\start\ge\e$. Together they imply that for at
least one $s\in\{a,b\}$ all the following hold:
\begin{itemize}
\item If $s$ overlaps $X[r+1]$, then $E[r+1]\endd-s\endd\ge\e$.
\item If $s$ overlaps $X[r-1]$, then $s\start-E[r-1]\start\ge\e$.
\end{itemize}
Therefore by \nameref{F.FA.Push.prop} (\ref{F.FA.Push.prop}), this
lemma follows directly.
\end{proof}
\clearpage{}
\begin{subr}[Push]
\label{F.FA.Push}~
\begin{block}
\item \textbf{\uline{Dependencies}}

\begin{block}
\item \textbf{\nameref{F.FA.Jump} (\ref{F.FA.Jump})}
\end{block}
\item \textbf{\uline{Implementation}}

\begin{block}
\item \textbf{Subroutine Push( insert state $(I=(n,T,W),S,r)$ , nat $m=\frac{2}{e}$
, slot $s$ ):}

\begin{block}
\item Define $X$ to be $S$ sorted with $(X[0]\endd,X[n+1]\start)=(-\infty,\infty)$.
\item // Sort the slots within $W[r]$ //
\item Let $u[1..q]$ be an ordered sequence such that $\class{S[u[i]]}{i\in[1..q]}=\class{S[i]}{i\in[1..n]\backslash\{r\}\land S[i]\wi W[r]}$.
\item Sort $S[u[1..q]]$.
\item // Allocate the inserted task to $s$ //
\item Set $S[r]=s$.
\item // Handle overlapping slots on both sides //
\item Set $jumps=\max\left(\log_{1+\frac{1}{2}e}\left(\frac{14}{e^{2}}\right),0\right)+1$.
\item For $i$ from $r+1$ up to $n$ as long as $S[i-1]$ overlaps $X[i]$:

\begin{block}
\item Set $i'\in[i..n]$ such that $S[i']=X[i]$.
\item // Jump if pushing $m$ slots is insufficient //
\item If $i=r+m+1$:

\begin{block}
\item Set $S[i']=null$.
\item If $\f{Jump}((I,S,i'),m,[S[r+m]\endd,W[i']\endd],\str{right},jumps)=Failure$:

\begin{block}
\item Return $Failure$.
\end{block}
\item Exit For.
\end{block}
\item // Swap into place and push aside the neighbouring slot //
\item Swap $S[i],S[i']$.
\item Set $S[i]=S[i-1]+1$.
\end{block}
\item For $i$ from $r-1$ down to $1$ as long as $S[i+1]$ overlaps $X[i]$:

\begin{block}
\item Set $i'\in[1..i]$ such that $S[i']=X[i]$.
\item // Jump if pushing $m$ slots is insufficient //
\item If $i=r-m-1$:

\begin{block}
\item Set $S[i']=null$.
\item If $\f{Jump}((I,S,i'),m,[W[i']\start,S[r-m]\start],\str{left},jumps)=Failure$:

\begin{block}
\item Return $Failure$.
\end{block}
\item Exit For.
\end{block}
\item // Swap into place and push aside the neighbouring slot //
\item Swap $S[i],S[i']$.
\item Set $S[i]=S[i+1]-1$.
\end{block}
\item Return $Success$.
\end{block}
\end{block}
\end{block}
\end{subr}
\begin{lem}[Push's properties]
\label{F.FA.Push.prop} Take any $\e$-underallocated ordered insert
state $(I=(n,T,W),S,r)$ and positive natural $m=\frac{2}{e}$. Let
$X$ be $S$ sorted with $(X[0]\endd,X[n+1]\start)=(-\infty,\infty)$.
Take any ordered $(1+\e)$-partial solution $E$ for $(I,r)$ and
slot $s$ such that all the following hold:
\begin{itemize}
\item $X[r-1]\le s\le X[r+1]$.
\item $s=Y[r]$ for some ordered solution $Y$ for $I$.
\item If $s$ overlaps $X[r+1]$, then $E[r+1]\endd-s\endd\ge\e$.
\item If $s$ overlaps $X[r-1]$, then $s\start-E[r-1]\start\ge\e$.
\end{itemize}
Then $\f{Push}((I,S,r),m,s)$ does the following:
\begin{itemize}
\item If $c<\frac{7}{e}+4$ and ( $m\ge n$ or $e\le\e$ ), it solves $(I,S,r)$
using at most $2\floor{\max\left(\log_{1+\frac{1}{2}\e}\left(\frac{14}{\e^{2}}\right),0\right)}+\frac{17}{e}+6$
reallocations and returns $Success$ in $O\left(\frac{1}{e}\log(\frac{1}{e})\log(n)+\log(n)^{2}\right)$
time.
\item If it returns $Failure$, it would have taken $O\left(\frac{1}{e}\log(\frac{1}{e})\log(n)\right)$
time.
\end{itemize}
\end{lem}
\begin{proof}
We shall divide the proof according to the parts of the subroutine.

\br

\textbf{High-level overview}

Push sorts all the slots in $S$ that are within the inserted window,
which ensures that all slots in earlier windows are before $s$ and
all slots in later windows are after $s$. Subsequently, it allocates
the inserted task to $s$, which may cause overlaps on both sides.
To fix that, it goes through up to $m$ slots on each side of $s$,
for each of them using a swap to align it with $X$ and then pushing
it aside so that it no longer overlaps the previous slot. If after
$m$ slots the last pushed slot $S[i]$ still overlaps the next one
$S[j]$, it deallocates $S[j]$ and executes Jump with starting window
$W[j]$ and starting span the subinterval of $W[j]$ that is beyond
$S[i]$. If the stipulated conditions are met, the starting span will
be large enough for Jump to succeed.

\br

\textbf{Symmetry}

It suffices to analyze the first for-loop since the second for-loop
and all the lemma conditions are symmetrical about $r$.

\br

\textbf{Sorting part}

First we shall consider the situation just after sorting the slots
within $W[r]$. Since the number of windows before $W[r]$ is the
same as the number of slots before $X[r]$, we have:
\begin{block}
\item $\#(\class i{i\in[1..r-1]\land S[i]>X[r-1]})$
\item $=(r-1)-\#(\class i{i\in[1..r-1]\land S[i]\le X[r-1]})$
\item $=\#(\class i{i\in[1..n]\backslash\{r\}\land S[i]\le X[r-1]})-\#(\class i{i\in[1..r-1]\land S[i]\le X[r-1]})$
\item $=\#(\class i{i\in[r+1..n]\land S[i]\le X[r-1]})$.
\end{block}
But for any $i\in[1..r-1]$ and $j\in[r+1..n]$ such that $S[i]>X[r-1]$
and $S[j]\le X[r-1]$, we would have $S[j]<S[i]$ and $S[i]\endd\le W[i]\endd\le W[r]\endd$
and $S[j]\start\ge W[j]\start\ge W[r]\start$, which imply $S[i],S[j]\wi W[r]$
and hence contradict the fact that $S[u[1..q]]$ is sorted. Therefore
we must have the counting identity $\#(\class i{i\in[1..r-1]\land S[i]>X[r-1]})$
$=\#(\class i{i\in[r+1..n]\land S[i]\le X[r-1]})=0$.

\br

\textbf{Pushing part}

Next we shall prove that for the first for-loop in the pushing part
the following invariances hold before each iteration:
\begin{enumerate}
\item $\#(\class j{j\in[1..i-1]\land S[j]>X[i-1]})=0$.
\item $S[j]=X[j]=S[r]+(j-r)$ for any $j\in[r..i-1]$.
\item $S[i']=X[i]$ for some $i'\in[i..n]$.
\item $S[r..n]$ is a valid allocation of $T[r..n]$ except possibly that
$S[i-1]$ overlaps $X[i]$.
\end{enumerate}
Invariance~1 holds by Invariance~2, which holds by construction,
and thus Invariance~3 holds since either $X[i]>X[i-1]$ or $X[i]=X[i-1]=S[i-1]>S[j]$
for any $j\in[1..i-2]$. Invariance~4 follows from Invariance~3,
because the swap will not cause any additional violation of allocation
validity by choice of $s$, and because $S[i]$ will be shifted to
$S[i-1]+1=X[i-1]+1\le X[i+1]$ and so will at most overlap $X[i+1]$.

\br

\textbf{Jumping part}

Finally if Jump is executed, $m<n$ otherwise all slots would have
been pushed, and thus the following inequalities hold:
\begin{block}
\item $E[r+m+1]\endd-S[r+m]\endd$ $\ge(E[r+1]\endd+m(1+\e))-(S[r]\endd+m)$
\item $=(E[r+1]\endd-S[r]\endd)+m\cdot\e$ $\ge(m+1)\e$ {[}by Invariance~1{]}.
\item $(m+1)\e-1=(\frac{2}{e}+1)\e-1\ge\frac{\e}{e}+\e>\max\left(1+\frac{1}{2}\e,\frac{(c-1)\e}{7}\right)$.
\end{block}
And so the lemma follows quite directly from \nameref{F.FA.Jump.prop}
(\ref{F.FA.Jump.prop}~Properties~3,4), where $d=7$. It should
be noted though that the state $(I,S,r)$ we are passing to Jump is
strictly speaking not an insert state, but it does not matter. To
verify that it still works, just prior to calling Jump we can pretend
modify $S[1..r-1]$ to make it an insert state due to the counting
identity and the fact that $s=Y[r]$ for some ordered solution $Y$
for $I$, and then Jump's properties hold. Afterward we just undo
our pretend modification, which does not interfere with Jump as a
consequence of Property~4.

\br

\textbf{Total costs}

The number of reallocations made by the sorting and pushing parts
is bounded by $\max(c,m)+3m\le\frac{13}{e}+4$, and hence the total
number of reallocations is at most $2\floor{\max\left(\log_{1+\frac{1}{2}\e}\left(\frac{14}{\e^{2}}\right),0\right)}+\frac{17}{e}+6$.
The time taken by sorting and pushing is obviously $O((c+2m)\log(n))\wi O\left(\frac{1}{e}\log(n)\right)$,
which is dominated by the time taken by Jump.
\end{proof}
\clearpage{}
\begin{subr}[Jump]
\label{F.FA.Jump}~
\begin{block}
\item \textbf{\uline{Implementation}}

\begin{block}
\item \textbf{Subroutine Jump( insert state $(I=(n,T,W),S,r)$ , nat $m=\frac{2}{e}$
, interval $U$ , string $dir$ , nat $jumps$ ):}

\begin{block}
\item Set $V=U$.
\item For $jmp$ from $0$ up to $jumps$:

\begin{block}
\item // Enter the final phase if $V$ has sufficient empty space //
\item Subroutine count( interval $w$ ):

\begin{block}
\item Return $\#(\class i{i\in[1..n]\backslash\{r\}\land S[i]\wi w})$.
\end{block}
\item Subroutine space( interval $w$ ):

\begin{block}
\item Set $sstart=\max\left(w\start,\max_{i:i\in[1..n]\backslash\{r\}\land S[i]\start<w\start}S[i]\endd\right)$.
\item Set $send=\min\left(w\endd,\min_{i:i\in[1..n]\backslash\{r\}\land S[i]\endd>w\endd}S[i]\start\right)$.
\item Return $(send-sstart)-\f{count}(w)$.
\end{block}
\item Set $blocks=\ceil{\frac{\f{count}(V)+1}{m+1}}$.
\item If $\f{space}(V)\ge blocks$:

\begin{block}
\item // Find a block of at most $m+1$ gaps within $V$ having empty space
at least $1$ //
\item Binary search to find interval $R$ such that all the following hold:

\begin{itemize}
\item $R\wi V$.
\item $\f{count}(R)\le m$.
\item $\f{space}(R)\ge1$.
\item $S[i]\wi R$ for any $S[i]$ that overlaps $R$.
\end{itemize}
\item // Sort the slots within $R$ //
\item Let $q=\f{count}(R)$.
\item Let $u[1..q]$ be an ordered sequence such that $\class{u[i]}{i\in[1..q]}=\class i{i\in[1..n]\backslash r\land S[i]\wi R}$.
\item Sort $S[u[1..q]]$.
\item // Pack the slots aside within $R$ to leave a gap $G$ of length
$1$ //
\item Set $i'=q+1$.
\item Set $G=R\endd+[-1,0]$.
\item For $i$ from $1$ up to $q$:

\begin{block}
\item Set $x=\max(R\start+(i-1),S[u[i]]\start-1)$.
\item If $x<W[u[i]]\start$:

\begin{block}
\item Set $i'=i$.
\item Set $G=x+[0,1]$.
\item Exit for.
\end{block}
\item Set $S[u[i]]=x+[0,1]$.
\end{block}
\item For $i$ from $i'$ up to $q$:

\begin{block}
\item Set $x=\min(R\endd-(q-i),S[u[i]]\endd+1)$.
\item Set $S[u[i]]=x+[-1,0]$.
\end{block}
\item // Reallocate along cascade from $W[r]$ to $G$ //
\item If $dir=\str{both}$:

\begin{block}
\item Set $dir=\cond{G\start<U\start}{\str{left}}{\str{right}}$.
\end{block}
\item Set $j=r$.
\item While $G\nwi W[j]$:

\begin{block}
\item Define $jslots=\class i{i\in[1..n]\backslash\{r\}\land S[i]\wi W[j]}$.
\item Set $j'=\cond{dir=\str{left}}{\min(jslots)}{\max(jslots)}$.
\item Set $S[j]=S[j']$.
\item Set $j=j'$.
\end{block}
\item Set $S[j]=G$.
\item Return $Success$.
\end{block}
\item // Jump to the furthest window{[}s{]} whose slot is within $V$ //
\item Define $jslots=\class i{i\in[1..n]\backslash\{r\}\land S[i]\wi V}$.
\item If $jslots=\none$:

\begin{block}
\item Return $Failure$.
\end{block}
\item Set $rstart=W[\min(jslots)]\start$.
\item Set $rend=W[\max(jslots)]\endd$.
\item If $dir=\str{both}$:

\begin{block}
\item Set $V=[\min(rstart,V\start),\max(rend,V\endd)]$.
\end{block}
\item Otherwise:

\begin{block}
\item Set $V=\cond{dir=\str{left}}{[rstart,V\endd]}{[V\start,rend]}$.
\end{block}
\end{block}
\item // Return failure if too many jumps are used //
\item Return $Failure$.
\end{block}
\end{block}
\end{block}
\end{subr}
\begin{lem}[Jump's properties]
\label{F.FA.Jump.prop} Take any ordered insert state $(I,S,r)$
and positive natural $m=\frac{2}{e}$ and real $\e\ge e$. Take also
any interval $U\wi W[r]$ and $dir\in\{\str{left},\str{right},\str{both}\}$
and natural $jumps$. Then $\f{Jump}((I,S,r),m,U,dir,jumps)$ has
the following properties:
\begin{enumerate}
\item It solves $(I,S,r)$ using at most $m$ reallocations if all the following
hold:

\begin{itemize}
\item $m\ge n$.
\item $\f{span}(U)\ge n+1$.
\item $jumps\ge1$.
\end{itemize}
\item It solves $(I,S,r$) using at most $(jumps+m)$ reallocations if for
some $d>0$ all the following hold:

\begin{itemize}
\item $(I,S,r)$ is $\e$-underallocated.
\item $U=W[r]$.
\item $dir=\str{both}$.
\item $c\ge\frac{6}{\e}+3+\frac{c-1}{d}$.
\item $jumps\ge\max\left(\log_{1+\frac{1}{2}\e}\left(\frac{4d}{\e}\right),0\right)+1$.
\end{itemize}
\item It solves $(I,S,r$) using at most $(jumps+m)$ reallocations if for
some $d>0$ either of the following holds:

\begin{itemize}
\item For some $r'\in[1..r-1]$ and ordered $(1+\e)$-partial solution $E$
for $(I,r'-m)$, all the following hold:

\begin{itemize}
\item $\#(\class i{i\in[1..r']\land S[i]>S[r']})=0$.
\item $U=[S[r']\endd,W[r]\endd]$.
\item $dir=\str{right}$.
\item $E[r'+1]\endd-S[r']\endd-1\ge\max\left(1+\frac{1}{2}\e,\frac{(c-1)\e}{d}\right)$.
\item $jumps\ge\max\left(\log_{1+\frac{1}{2}\e}\left(\frac{2d}{\e^{2}}\right),0\right)+1$.
\end{itemize}
\item For some $r'\in[r+1..n]$ and ordered $(1+\e)$-partial solution $E$
for $(I,r'+m)$, all the following hold:

\begin{itemize}
\item $\#(\class i{i\in[r'..n]\land S[i]<S[r']})=0$.
\item $U=[W[r]\start,S[r']\start]$.
\item $dir=\str{left}$.
\item $E[r'-1]\start-S[r']\start+1\le-\max\left(1+\frac{1}{2}\e,\frac{(c-1)\e}{d}\right)$.
\item $jumps\ge\max\left(\log_{1+\frac{1}{2}\e}\left(\frac{2d}{\e^{2}}\right),0\right)+1$.
\end{itemize}
\end{itemize}
\item If it solves $(I,S,r)$, it also satisfies all the following:

\begin{itemize}
\item It performs reallocations to slots completely after $U\start$ if
$dir=\str{right}$ and before $U\endd$ if $dir=\str{left}$.
\item It returns $Success$ in $O\left(\frac{1}{e}\log(\frac{1}{e})\log(n)+\log(n)^{2}\right)$
time.
\end{itemize}
\item If it returns $Failure$, it would have taken $O\left(jumps\cdot\log(n)\right)$
time.
\end{enumerate}
\end{lem}
\begin{proof}
The proof is divided roughly according to the parts of the subroutine.

\br

\textbf{High-level overview}

Here is a high-level sketch of what Jump does. In the initial phase,
$k$ will be the current window index, and $V$ will be the current
span. $(k,V)$ are initialized based on the input parameters, and
jumps will be made only in the specified direction, each time to the
furthest window that has a slot completely within the current span.
After each jump $V$ will be extended in that direction to match the
reach of the new window. The final phase is entered when the current
span has sufficient total empty space to guarantee that it contains
a region $R$ containing at most $m$ slots and having empty space
of at least $1$, upon which the slots within $R$ will be sorted
and packed within $R$ to leave a unit gap $G$, and then the solution
is finished by reallocating the slots on the cascade of jumps from
$W[r]$ to $G$.

Jump succeeds in the specified situations roughly because of the following
high-level reasons. The fact that the original instance is $\e$-slack
before insertion implies that the windows stretch out at an average
rate of at least about $1+\e$, while in the initial phase the slots
stretch out at an average rate of at most about $1+\frac{1}{m}=1+\frac{e}{2}$.
So if $e\le\e$, the current span $V$ grows more or less exponentially.
All that the proof really depends on is that $m\in\Q(\frac{1}{e})$
and $\e-\frac{1}{m}\in\Q(\e)$, so the choice of $m=\frac{2}{e}$
is just to make the computations simple instead of attempting to obtain
optimal bounds.

\br

\textbf{Jumping sequence}

Before getting to the proof, here is the precise definition of jumps.
Call $v$ a \textbf{jumping sequence} iff the following hold:
\begin{itemize}
\item $v[1..\#(v)]$ is a strictly monotonic sequence of indices.
\item $S[v[2..\#(v)]]$ is a strictly monotonic sequence of slots.
\item $v\first=r$.
\item $W[v[i]]\co S[v[i+1]]$ for any $i\in[1..\#(v)-1]$.
\end{itemize}
Additionally, we say that $v$ covers a point $x$ iff $x\in\f{span(}W[v])$.
Note that a jumping sequence is allowed to have jumps that are not
to the furthest possible window. It is quite clear from the construction
of Jump that before any iteration of the main for-loop, for any $x\in V$
there is some jumping sequence $v$ of length at most $jmp+1$ that
covers $x$. Note that the only condition on $U$ needed for this
is $U\wi W[r]$, and it is sufficient to guarantee that a solution
will be found if the final phase is entered.

\br

\textbf{Failure characteristics}

Clearly each invocation of $\f{count}$ and $\f{space}$ takes $O(\log(n))$
time using only a few search queries to the IDSetRQ. Also, Jump does
not compute the whole set $jslots$ but just its minimum and maximum,
which takes only two search queries and one range query to the IDSetRQ,
which amounts to $O(\log(n))$ time for each jump. Therefore if Jump
returns failure it would have taken only $O(jumps\cdot\log(n))$ time
in total.

\br

\textbf{Final phase}

If $\f{span}(U)\ge n+1$ and $m\ge n$, then in the very first iteration
of the for-loop over $jmp$ we have $blocks\le\ceil{\frac{n+1}{m+1}}\le1$
and $\f{space}(V)\ge\f{span}(V)-n=\f{span}(U)-n\ge1$, and hence the
final phase will be entered, in which the sorting part and $U\wi W[r]$
ensures that the packing part succeeds, giving Property~1.

\br

\uline{Searching part}

If the final phase is entered, the binary search can be carried out,
because the intervening gaps between the slots within $V$ can be
partitioned into $\ceil{\frac{count(V)+1}{m+1}}$ blocks of $m+1$
or less, and $\f{space}$ is additive on intervals, and so the binary
search can keep halving the current set of consecutive blocks by choosing
the half that has the larger average space per block, which ensures
that it will obtain a single block $R$ with total space at least
$1$. It is also easy to guarantee that $S[i]\wi R$ for any $S[i]$
that overlaps $R$, by trimming $R$ to exclude any slots that cross
its boundary. The binary search takes $O(\log(blocks)\log(n))\wi O(\log(n)^{2})$
time.

\br

\uline{Sorting part}

$S[u[1..q]]$, which are the slots within $R$, is then sorted, which
takes $O(m\cdot\log(n))\wi O(\frac{1}{e}\log(n))$ time, and by \nameref{F.Ord}
(\ref{F.Ord}) $S$ remains a partial solution for $I$. Let $S_{0}$
be the original $S$ before the sorting. Take any shortest increasing
jumping sequence $v_{0}$ that covers $R\endd$ before the sorting.
After sorting, $v_{0}$ may be no longer a jumping sequence, but we
can create a new jumping sequence $v$ that covers $R\endd$ by modifying
$v_{0}$.

For each $i$ just before $j$ in $v_{0}$, we have $S[i]\endd<S[j]\endd<R\endd$
by minimality of $v_{0}$. If $S[j]\nwi R$, we have $S[j]\start<R\start$
and hence $S[i],S[j]$ are unchanged, so we do nothing. If $S[j]\wi R$,
we insert each element in $A=\class a{a\in u[1..q]\land i<a<j}$ into
$v$ such that $v$ remains strictly increasing.

Now consider any $b\in A\cup\{j\}$ and let $a$ be just before $b$
in $v$. If $a\in u[1..q]$, we have $S[a]<S[b]$ because $S[u[1..q]]$
is ordered. If $a\notin u[1..q]$, it must be that $a=i$ and so $S[a]\endd<R\endd$,
which gives $S[a]\start<R\start$ since $S[a]\nwi R$, and hence $S[a]<S[b]$.
Therefore in all cases $S[a]<S[b]$.

Also, since the number of windows with slots before $S[b]$ within
$R$ is the same as the number of windows before $W[b]$ with slots
within $R$, we have:
\begin{block}
\item $\#(\class k{k\in u[1..q]\land k\ge b\land S_{0}[k]<S[b]})$
\item $=\#(\class k{k\in u[1..q]\land S_{0}[k]<S[b]})-\#(\class k{k\in u[1..q]\land k<b\land S_{0}[k]<S[b]})$
\item $=\#(\class k{k\in u[1..q]\land k<b})-\#(\class k{k\in u[1..q]\land k<b\land S_{0}[k]<S[b]})$
\item $=\#(\class k{k\in u[1..q]\land k<b\land S_{0}[k]\ge S[b]})$.
\end{block}
And so if $W[a]\endd<S[b]\endd$, for any $k\in u[1..q]$ such that
$k<b$ we have $k\le a$ and so $S_{0}[k]\endd\le W[k]\endd\le W[a]\endd<S[b]\endd$,
and hence $W[a]\endd\ge W[i]\endd\ge S_{0}[j]\endd\ge S[b]\endd$
since $j\in u[1..q]$ and $j\ge b$. Therefore in all cases $W[a]\endd\ge S[b]\endd$.

Therefore, after all the insertions, $v$ is once more an increasing
jumping sequence that covers $R\endd$. Also, all the elements added
to $v$ are in $u[1..q]$.

\br

\uline{Packing part}

Next is the packing part, in which the first for-loop does left-packing
and the second for-loop does right-packing. Call each packing successful
iff $S$ is still a valid allocation after it is completed. First
note that since $S[u[1..q]]$ is ordered, by induction $R\start+(i-1)\le S[u[i]]\start$
for any $i\in[1..q]$, and likewise $R\start-(q-i)\ge S[u[i]]\endd$
for any $i\in[1..q]$. Thus the first for-loop sets $x\le S[u[i]]\start$
and hence sets $S[u[i]]\wi W[u[i]]$, for each $i\in[1..i'-1]$. Also,
it sets $S[u[1..i'-1]]$ to non-overlapping slots, because $\max(R\start+(i-1),S[u[i]]\start-1)+1\le\max(R\start+i,S[u[i+1]]\start-1)$
for any $i\in[1..q-1]$. Since $S[u[1..i'-1]]$ are also all still
within $R$ and not shifted right, they do not overlap other slots
in $S$, and hence the left-packing always succeeds. Likewise, the
right-packing succeeds if the second for-loop always sets $S[u[i]]\endd\le W[u[i]]\endd$
for each $i\in[i'..q]$, which is what we shall now prove.

If the right-packing fails and sets $S[i]\endd>W[i]\endd$ for some
$i\in u[i'..q]$, then it suffices to consider the case that $i\ge r$
because this asymmetric part works if and only if the symmetric version
that tries both directions works. Basically, since the symmetric version
guarantees some way of packing that works, the asymmetric version
left-packs at least as many slots as the symmetric version, and hence
can definitely right-pack the rest. Anyway this is inconsequential
and so details are omitted.

Let $v$ be a shortest increasing jumping sequence that covers $R\endd$,
which exists by the earlier argument because $R\endd\ge S[i]\endd>W[i]\endd\ge W[r]\endd$.
Then $v\first\le i<v\last$ since $W[v\last]\endd\ge R\endd\ge S[i]\endd>W[i]\endd$.
Let $h$ be just before $j$ in $v$ such that $h\le i<j$, and let
$S_{1}$ be the original $S$ before packing. Then $S_{1}[j]\start\ge W[j]\start\ge W[i]\start\ge W[u[i']]\start>R\start$
and $S_{1}[j]\endd\le S_{1}[v\last]\endd<R\endd$ by minimality of
$v$, and hence $S_{1}[j]\wi R$. Thus $S_{1}[i]<S_{1}[j]$ since
$S_{1}[u[1..q]]$ is ordered. This gives $S[i]\endd\le S_{1}[i]\endd+1\le S_{1}[j]\endd\le W[h]\endd\le W[i]\endd$,
which contradicts the failure of the right-packing.

Therefore the packing always succeeds, and sets $G$ such that $G\wi R$
and $G$ does not overlap any slot in $S$, because either $i'\le q$
and $G\endd=\max(R\start+(i'-1),S[u[i']]\start-1)+1$ $\le\min(R\endd-(q-i'),S[u[i']]\endd+1)-1=S[u[i']]\start$,
or $i'=q+1$ and $G\start=R\endd-1$ $\ge\max(R\start+(q-1),S[u[q]]\start-1)+1=S[u[q]]\endd$.

\br

\uline{Cascading part}

Next, we shall prove that the cascade at the end finishes the modification
of $S$ to a solution for $I$. Just as in the main for-loop, this
while-loop makes $j$ trace out some jumping sequence, but we will
not actually need to prove so much. Let $f[l+1]$ be the value of
$j$ after $l$ iterations. It suffices to handle the case that $dir=\str{right}$.
Note that $f$ is increasing, because either it has only one element,
or $G\nwi U$ and so $G\endd>U\endd$ since there is some increasing
jumping sequence $v$ that covers $R\endd$, in which case $\neg(G\start<U\start)$
since $\f{span}(U)\ge1$. After each iteration of the while-loop,
$S$ is a solution for $I$ except for an exact overlap at $S[j]$,
so Jump is done if the while-loop terminates. To check that this always
occurs, we shall prove that $v[l+1]\le f[l+1]$ after $l$ iterations
of the while-loop as follows.

After $0$ iterations, the invariance holds trivially because $v[1]=r=f[1]$.
After $l$ iterations where $l>0$, let $l'=\max(\class i{i\in[1..\#(v)]\land v[i]\le f[l]})\ge l$
by the invariance. Then by the monotonicity of $v,f$, $\f{span}(W[f[1..l]])\co\f{span}(W[v[1..l']])$.
Also, $G\wi R\wi\f{span}(W[v])$, but $G\nwi\f{span}(W[v[1..l']])$
because $G\nwi\f{span}(W[f[1..l]])$ by the while-loop condition.
Thus $l'<\#(v)$ and so $v[l']\le f[l]<v[l'+1]$ and $W[v[l']]\co S[v[l'+1]]$,
which gives $W[f[l]]\co[W[v[l'+1]]\start,W[v[l']]\endd]\co S[v[l'+1]]$,
and hence $f[l+1]=\max(\class i{i\in[1..n]\land S[i]\wi W[f[l]]})$
$\ge v[l'+1]\ge v[l+1]$.

Therefore the while-loop runs for at most $(\#(v)-1)$ iterations,
otherwise $G\wi\f{span}(W[v])\wi\f{span}(W[f[1..\#(v)]])$, contradicting
the while-loop condition on iteration $\#(v)$. Thus the cascading
part makes at most $(\#(v)-1)$ reallocations, of which at most $\#(v_{0})-1=jmp\le jumps$
are to slots outside $R$.

\br

\textbf{Total costs}

In total Jump makes at most $(m+jumps)$ reallocations and takes $O\left(\log(n)^{2}+m\cdot\log(n)+jumps\cdot\log(n)\right)$
time.

\br

\textbf{Initial phase}

What remains to establish the lemma is to prove that the final phase
is entered in either of the two specified situations. So we shall
look at the jumping part of the initial phase, just before $V$ is
modified. Let $V'$ be the new value of $V$ after the modification,
and let $s=\f{span}(V)$ and $s'=\f{span}(V')$ and $z=\f{count}(V)$.

\br

\uline{Property 2}

Next, consider Property~2 as specified in the lemma statement (\ref{F.FA.Jump.prop}).
Then $s-(z+2)<\f{space}(V)<blocks=\ceil{\frac{z+1}{m+1}}=\floor{\frac{z}{m+1}}+1$.
Thus $s<z+\frac{z}{m+1}+3=z\frac{m+2}{m+1}+3$, and hence $z>(s-3)\frac{m+1}{m+2}=(s-3)\frac{2+e}{2+2e}\ge(s-3)\frac{2+\e}{2+2\e}\ge0$
since $s\ge c\ge3$. Thus $z>0$ and so $jslots=\class i{i\in[1..n]\backslash\{r\}\land S[i]\wi V}\ne\none$.
Let $a=\min(jslots)$ and $b=\max(jslots)$. Then $\#([a..b])\ge z$,
which implies that $s'\ge\f{span}(W[a..b])\ge z(1+\e)$.

From these we get $s'>(s-3)\frac{2+\e}{2+2\e}(1+\e)=(s-3)(1+\frac{1}{2}\e)$,
which gives $s'-(\frac{6}{\e}+3)>\left(s-(\frac{6}{\e}+3)\right)(1+\frac{1}{2}\e)$,
and by induction $s-(\frac{6}{\e}+3)\ge\left(c-(\frac{6}{\e}+3)\right)(1+\frac{1}{2}\e)^{jmp}\ge\frac{c-1}{d}(1+\frac{1}{2}\e)^{jmp}$.
As before, $\left(s-(\frac{6}{\e}+3)\right)\frac{1}{2}\e<s'-s\le2(c-1)$
because $V'\endd=W[b]\endd\le S[b]\endd+(c-1)\le V\endd+(c-1)$ and
similarly $V'\start\ge V\start-(c-1)$, and hence $s-(\frac{6}{\e}+3)<\frac{4(c-1)}{\e}$.
Therefore $jmp\le\log_{1+\frac{1}{2}\e}\left(\frac{4d}{\e}\right)$,
and so the for-loop will always enter the final phase since $jumps\ge\max\left(\log_{1+\frac{1}{2}\e}\left(\frac{4d}{\e}\right),0\right)+1$.

\br

\uline{Property 3}

First, consider Property~3 as specified in the lemma statement (\ref{F.FA.Jump.prop}).
By symmetry it suffices to consider the case of $dir=\str{right}$.
Since $V\start=U\start$ and no slot straddles $U\start$, $s-(z+1)<\f{space}(V)<blocks=\ceil{\frac{z+1}{m+1}}=\floor{\frac{z}{m+1}}+1$.
Thus $s<z+\frac{z}{m+1}+2=z\frac{m+2}{m+1}+2$, and hence $z>(s-2)\frac{m+1}{m+2}=(s-2)\frac{2+e}{2+2e}$.
We have $U\endd=W[r]\endd\ge E[r]\endd\ge E[r'+1]\endd+(r-r'-1)(1+\e)$,
which gives $s-1=\f{span}(V)-1\ge\f{span}(U)-1$ $\ge(r-r'-1)(1+\e)+E[r']\endd-U\start-1$
$\ge(r-r'-1)(1+e)+\max(1+\frac{1}{2}\e,\frac{(c-1)\e}{d})$ and so
$s\ge2$. Thus $z>r-r'-1\ge0$ and so $jslots=\class i{i\in[1..n]\backslash\{r\}\land S[i]\wi V}\ne\none$,
and hence $\max(jslots)\ge r'+z+1$ since $jslots\wi[r'+1..n]\backslash\{r\}$,
which implies that $rend\ge E[\max(jslots)]\endd\ge E[r'+1]\endd+z(1+\e)$.

From these we get the following inequality:
\begin{block}
\item $s'-1=rend-U\start-1$
\item $\ge z(1+\e)+(E[r'+1]\endd-U\start-1)$
\item $>(s-2)\frac{2+e}{2+2e}(1+\e)+(1+\frac{1}{2}\e)$
\item $\ge(s-2)\frac{2+\e}{2+2\e}(1+\e)+(1+\frac{1}{2}\e)$
\item $=(s-1)(1+\frac{1}{2}\e)$.
\end{block}
Thus by induction $s-1\ge\frac{(c-1)\e}{d}(1+\frac{1}{2}\e)^{jmp}$.
Furthermore, rearranging the inequality results in $(s-1)\frac{1}{2}\e<s'-s\le c-1$
because $V'\endd=W[\max(jslots)]\endd\le S[\max(jslots)]\endd+(c-1)\le V\endd+(c-1)$,
and hence $s-1<\frac{2(c-1)}{\e}$. Therefore $jmp\le\log_{1+\frac{1}{2}\e}\left(\frac{2d}{\e^{2}}\right)$,
and so the for-loop will always enter the final phase since $jumps\ge\max\left(\log_{1+\frac{1}{2}\e}\left(\frac{2d}{\e^{2}}\right),0\right)+1$.

\br

\textbf{Conclusion}

It is now trivial to verify the remaining properties that we have
not explicitly justified. Therefore the lemma is proven at last.
\end{proof}
\clearpage{}

\section{Variable window length}

In this section we now turn to the variant of the problem that allows
variable window lengths. We shall first show that in the case of sufficiently
small slack there is no efficient allocator. Subsequently we present
for the case of sufficiently large underallocation an allocator that
we conjecture to be efficient, but we have been unable to prove it.
There are also a number of unanswered questions that do not fall under
either of these two cases.

\subsection{Small slack\label{V.SmallSlack}}

Even if the windows must have bounded relative ratio, if the slack
is sufficiently small, it is not difficult to come up with a sequence
of operations that requires $\W(n)$ reallocations per operation on
average with at most $n$ tasks in the system at any time. One example
is as follows. Insert tasks with windows $[2i,2i+3]\g$ and $[2i+1,2i+2]\g$
for each $i\in[1..k]$ on 1 processor where $\g=1+\e$. Then repeatedly
insert and delete a task with window alternating between $[2,3]\g$
and $[2k+2,2k+3]\g$. The instances are all clearly $\e$-slack with
relative window ratio at most $3$, but require $k$ reallocations
for each alternation for any $\e\in(0,1/3)$.

In general, if the relative window ratio is allowed to be at most
$1+\frac{2}{m}$ for some positive integer $m$, there is a pair of
$\e$-slack instances with $(m+1)k+1$ tasks along the same lines
such that alternating between them requires $k$ reallocations per
operation for any $\e\in\left(0,\frac{1}{2m+1}\right)$, by having
$k$ groups of $m+1$ windows, each group having $m$ windows of length
$m\g$ centred with 1 window of length $(m+2)\g$, the larger windows
in adjacent groups overlapping by $\g$, where $\g=1+\e$. For $p$
processors the situation is no better, since the above instances can
be just multiplied into $p$ copies, which require $\W(\frac{n}{(m+1)p})$
reallocations per operation.

This answers the question posed by Bender et al. in the negative,
but it leads to others: If $\e\ge\frac{1}{3}$, or if $\e\in(0,1)$
and the relative window ratio is at most $1+\frac{4\e}{1-\e}$, is
there an allocator such that the number of reallocations it makes
on each insertion depends only on the slack $\e$? If one exists,
it definitely cannot work on arbitrary feasible insert states because
it is easy to construct some that necessitate $\W(\log(n))$ reallocations
as we shall see in the next section.

\subsection{Large underallocation\label{V.LargeUnderalloc}}

For large underallocation, it seems that it is more natural to describe
$\e$-slack instances as $\g$-underallocated where $\g=1+\e$. Also,
a useful special type of instance is one with aligned windows, windows
whose endpoints are consecutive powers of $2$. The reason is that
aligning all the windows of a $4\g$-underallocated instance gives
a $\g$-underallocated instance for any $\g$ that is a power of $2$,
so for sufficiently large underallocation it reduces to solving the
problem for the aligned instance.

Unlike the fixed window variant, for any allocator $A$, no matter
how large $\g$ is, there is some sequence of operations comprising
just insertions such that the instance is always both aligned and
$\g$-underallocated but requires at least one reallocation, and furthermore
there is a $\g$-underallocated aligned insert state that requires
$\W\left(\frac{\log(n)}{\log(\g)}\right)$ reallocations where $n$
is the current number of tasks in the system. The latter also means
that any allocator that takes $o(\log(n))$ reallocations cannot work
on generic insert states and must avoid such bad insert states.

The rest of this section contain the precise statements and proofs
of the above theorems. Finally, we set forth an allocator VA and our
hypothesis that it takes only at most $1$ reallocation per insertion
as long as the aligned instance is always $2$-underallocated.
\begin{defn}[Aligned interval]
 Call an interval \textbf{aligned }iff it is $[a,a+1]2^{k}$ for
some $k,a\in\zz$.
\end{defn}

\begin{defn}[Align]
\label{V.Align} For any interval $X$, let $\f{\mathbf{align}}(X)$
be some largest aligned interval within $X$. Also, whenever we say
that we \textbf{align} $X$ we mean that we change it to $\f{align}(X)$.
\end{defn}

\begin{defn}[Interval's halves]
 For any interval $X$, let $\f{left}(X)$ and $\f{right}(X)$ be
the left half of $X$ and right half of $X$ respectively.
\end{defn}

\begin{defn}[Aligned interval's parent]
 For any aligned interval $X$, let $\f{parent}(X)$ be the unique
aligned interval of length $2\f{span}(X)$ that contains $X$.
\end{defn}

\begin{defn}[Aligned interval's sibling]
 For any aligned interval $X$, let $\f{sibling}(X)$ be the unique
aligned interval $Y$ such that $\f{parent}(X)=\f{parent}(Y)$ and
$X\ne Y$.\end{defn}
\begin{thm}[Align's properties]
\label{V.Align.prop} Take any aligned interval $X$. Then there
is some interval $Y$ with length $4\f{span}(X)$ such that ( $Y\co Z$
for any interval $Z$ such that $\f{align}(Z)\wi X$ ).\end{thm}
\begin{proof}
Let $k,a\in\zz$ such that $X=[a2^{k},(a+1)2^{k}]$. By symmetry we
can assume that $a$ is even. Let $Y=[(a-2)2^{k},(a+2)2^{k}]$. Then
$\f{span}(Y)=4\f{span}(X)$. Now take any interval $Z$ such that
$\f{align}(Z)\wi X$. Then $Z\start>(a-2)2^{k}$ otherwise $\f{align}(Z)$
is smaller than the aligned window $[(a-2)2^{k},a2^{k}]\wi Z$. Similarly
if $\f{align}(Z)=X$, we have $Z\endd<(a+2)2^{k}$ otherwise $\f{align}(Z)$
is smaller than $[a2^{k},(a+2)2^{k}]\wi Z$. Finally, if $\f{align}(Z)\ne X$,
we again have $Z\endd<(a+2)2^{k}$ otherwise $\f{align}(Z)$ is smaller
than $[(a+1)2^{k},(a+2)2^{k}]\wi Z$. Therefore in all cases $Z\wi Y$.\end{proof}
\begin{thm}[Alignment Reduction]
\label{V.Alignment} Take any $k\in\zz$ and $\g=2^{k}$ and $4\g$-underallocated
instance $I$. Let $J$ be $I$ with all windows aligned. Then $J$
is $\g$-underallocated.\end{thm}
\begin{proof}
Since $J$ has finitely many tasks, we can recursively allocate the
tasks by structural induction on aligned intervals, the invariance
being that all tasks with smaller windows can be allocated to aligned
slots. At each step, take some smallest aligned interval $W$ with
length $x$ such that at least one unallocated task in $J$ has window
$W$. Then $W$ is either identical to or disjoint from the window
for every unallocated task. By \nameref{V.Align.prop} (\ref{V.Align.prop}),
all the windows in $I$ that are aligned to within $W$ in $J$ are
within an interval of length $4x$, and hence there are at most $\frac{4x}{4\g}$
such windows since $I$ is $4\g$-underallocated. Thus there are at
most $\frac{x}{\g}$ tasks in $J$ with window within $W$, and we
can allocate those of them that are currently unallocated to aligned
slots by the invariance and since $\g$ is a power of $2$. Therefore
$J$ has a $\g$-solution.\end{proof}
\begin{rem*}
The condition in \nameref{V.Alignment} that $\g$ is a power of $2$
cannot be omitted as the theorem is false if $\g=7$. A counterexample
is an instance with $9$ tasks where $8$ of them have window $[1,2^{8}-1]$
and one has window $2^{6}+4\cdot2^{3}+[-14,14]$, which can align
to $2^{6}+[0,8\cdot2^{3}]$ and $2^{6}+[3\cdot2^{3},4\cdot2^{3}]$
respectively. The unaligned instance is $4(7)$-underallocated because
$9\cdot4(7)<2^{8}-2$, but the aligned instance is not $7$-underallocated,
because $2^{6}+[0,4\cdot2^{3}]$ and $2^{6}+[3\cdot2^{3},8\cdot2^{3}]$
can accomodate only $4$ and $5$ slots of length $7$ respectively,
since $5\cdot7>4\cdot2^{3}$ and $6\cdot7>5\cdot2^{3}$.\end{rem*}
\begin{defn}[Density]
 Take any allocator $A$. At any point in time, let $I$ be the current
instance and $S$ be the set of slots in the current solution maintained
by $A$. For any interval $X$, let $\f{density}(X)=\frac{1}{\f{span}(X)}\sum_{s\in S}\f{span}(s\cap X)$.\end{defn}
\begin{thm}[Reallocation Requirement]
\label{V.ReallocReq} Take any allocator $A$ and $\g\in\nn^{+}$.
Let $I$ be the current instance and $S$ be the set of slots in the
current solution maintained by $A$. Then there is some sequence of
$2^{2\g-1}$ insertions such that $I$ is always $\g$-underallocated
but $A$ makes at least one reallocation.\end{thm}
\begin{proof}
If $A$ does not perform any reallocation for any insertion sequence,
we do the following. Set $W_{1}=[0,\g\cdot2^{2\g-1}]$. We shall now
construct the insertion sequence inductively such that the following
invariances hold after step $k$ for each $k$ from $0$ to $2\g-1$:
\begin{enumerate}
\item $\f{span}(W_{k+1})=\g\cdot2^{2\g-(k+1)}$.
\item $\f{density}(W_{k+1})\ge\frac{k}{2\g}$.
\item $I$ is $\g$-underallocated.
\end{enumerate}
At step $k$ from $1$ to $2\g-1$, insert $2^{2\g-k-1}$ tasks all
with window $W_{k}$. Let $A=\f{left}(W_{k})$ and $B=\f{right}(W_{k})$.
Then by Invariances~1,2 $\f{density}(W_{k})\ge\frac{k-1}{2\g}$ before
these insertions, and hence $\f{density}(W_{k})\ge\frac{k-1}{2\g}+\frac{2^{2\g-k-1}}{\g\cdot2^{2\g-k}}\ge\frac{k}{2\g}$
after the insertions. Thus $\max(\f{density}(A),\f{density}(B))\ge\frac{k}{2\g}$.
Set $W_{k+1}\in\{A,B\}$ such that $\f{density}(W_{k+1})\ge\frac{k}{2\g}$.
Therefore Invariances~1,2 are preserved. Note that $I$ is still
$\g$-underallocated, because all the tasks with windows $W_{i}$
can be allocated to $\g$-length slots within $W_{i}\backslash W_{i+1}$
for each $i\in[1..k-1]$, giving Invariance~3.

After $(2\g-1)$ steps, we have inserted $(2^{2\g-1}-1)$ tasks and
$\f{density}(W_{2\g})\ge1$. Insert one more task with window $W_{2\g}$,
which is possible because $I$ is still $\g$-underallocated by the
same argument as before. Now $\f{density}(W_{2\g})>1$, which is impossible.

Therefore the theorem follows.\end{proof}
\begin{rem*}
Note that \nameref{V.ReallocReq} holds for any $\g\in\nn^{+}$, but
if $\g$ is also a power of $2$, the inserted windows will be aligned,
and hence reallocations are necessary even if the instance is always
aligned.\end{rem*}
\begin{thm}[Underallocation Requirement]
\label{V.UnderallocReq} Take any allocator $A$. Then there is some
sequence of operations that insert only tasks with aligned windows
such that $A$ makes $\log(n)$ reallocations on every subsequent
insertion after an initial segment of the sequence, where $n$ is
the number of tasks in the system.\end{thm}
\begin{proof}
Take any $k\in\nn$. In the setup phase, insert one task with window
$X$ for every aligned interval $X$ within $[0,2^{k}]$ of length
at least $2$. This is clearly feasible and the setup phase is complete.
Let $I=(n,T,W)$ be the current solution. Start from the interval
$C=[0,2^{k}]$. While $C$ has length at least $2$, we have $C=W[i]$
for some $i$, and so update $C$ to the half that overlaps $S[i]$.
Then if a new task is inserted with window within the new $C$, $S[i]$
must be reallocated. This makes $C$ trace a path of $k+1$ distinct
intervals. At the end $C$ is an interval of length $1$, so insert
a task with window $C$. It is not too hard to check by induction
that $I$ is feasible and now $n=2^{k}$, and hence $A$ will have
to make $k=\log(n)$ reallocations. After that, delete the task with
window $C$, upon which there is yet again exactly one task with window
$X$ for every aligned interval $X$ within $[0,2^{k}]$ of length
at least $2$. Repeating then yields the theorem.\end{proof}
\begin{rem*}
Note that \nameref{V.UnderallocReq} holds even if reallocations are
allowed on deletions, since the same proof works.\end{rem*}
\begin{thm}[Non-Genericity]
\label{V.NonGeneric} Take any $\g=2^{k}$ for some $k\in\nn$. Then
there is some insert state $(I,S,r)$ with aligned $\g$-underallocated
$I$ that requires at least $\log_{2\g}(n-1)$ reallocations to be
solved.\end{thm}
\begin{proof}
Take any $m\in\nn$. Let $X_{i}=[0,(2\g)^{i}]$ for each $i\in[1..m]$,
and define an insert state $I=(n,T,W)$ as follows:
\begin{block}
\item $n=(2\g)^{m}+1$.
\item $W[i]=X_{j+1}$ for each $i\in[(2\g)^{j-1}+1..(2\g)^{j}]$ for each
$j\in[1..m]$.
\item $W[n]=X_{1}$.
\item $S[i]=i+[-1,0]$ for each $i\in[1..n-1]$.
\item $S[n]=null$.
\end{block}
Then $I$ is $\g$-underallocated, because it has a $\g$-solution
$G$ where $G[i]=i(2\g)+[-\g,0]$ for each $i\in[1..n-1]$ and $G[n]=[0,\g]$,
since $G[i]\endd=i(2\g)\le(2\g)^{j+1}=X_{j+1}\endd=W[i]\endd$ for
each $i\in[(2\g)^{j-1}+1..(2\g)^{j}]$ for each $j\in[1..m]$. Also,
$X_{j}$ contains $(2\g)^{j}$ slots for each $j\in[1..m]$ since
$X_{j}\endd=(2\g)^{j}\le n-1$.

Now take any sequence of reallocations that modifies $S$ to a solution
$S'$ for $I$. Clearly at least one slot within $X_{j}$ must be
reallocated from inside to outside $X_{j}$ for each $j\in[1..m]$
since there are now too many slots within $X_{j}$, and for such a
reallocated slot $S[i]$ it must be that $W[i]=X_{j+1}$, because
$W[i]$ is larger than $X_{j}$ but not larger than $X_{j+1}$ since
originally $i\le(2\g)^{j}$. Therefore these slots that must be reallocated
are distinct for distinct $j$, which implies that solving $(I,S,r)$
requires reallocating at least $m=\log_{2\g}(n-1)$ slots.\end{proof}
\begin{algo}[VA]
\label{V.VA}~
\begin{block}
\item \textbf{\uline{Variables}}

\begin{block}
\item Ordered instance $I=(n,T,W)$ // current aligned instance ; must be
feasible before and after each operation
\item Allocation $S$ // current allocation for $I$ ; must be a solution
for $I$ before and after each operation
\end{block}
\item \textbf{\uline{Initialization}}

\begin{block}
\item Set $I=(n,T,W)=(0,(),())$ and $S=()$.
\end{block}
\item \textbf{\uline{External Interface}}

\begin{block}
\item \textbf{Procedure Insert( task $t$ , window $w$ )} // inserts task
$t$ with window $w$ into the system
\item \textbf{Procedure Delete( task $t$ )} // deletes task $t$ from the
system
\end{block}
\item \textbf{\uline{Implementation}}

\begin{block}
\item \textbf{Subroutinecount( aligned interval $X$ ):}

\begin{block}
\item Return $\#(\class i{S[i]\wi X})$.
\end{block}
\item \textbf{Subroutine high( aligned interval $X$ ):}

\begin{block}
\item Return $\max(\class{\f{span}(W[i])}{S[i]\wi X})$.
\end{block}
\item \textbf{Subroutine best( aligned interval $X$ ):}

\begin{block}
\item If $\f{span}(X)=1$:

\begin{block}
\item Return $X$.
\end{block}
\item Set $A=\f{left}(X)$.
\item Set $B=\f{right}(X)$.
\item Return $\f{best}\cond{\f{count}(A)<\f{count}(B)}A{\cond{\f{count}(A)>\f{count}(B)}B{\cond{\f{high}(A)\ge\f{high}(B)}AB}}$.
\end{block}
\item \br\textbf{Subroutine bad( aligned interval $X$ ):}

\begin{block}
\item If $\f{span}(X)=1$:

\begin{block}
\item Return $X$.
\end{block}
\item Set $A=\f{left}(X)$.
\item Set $B=\f{right}(X)$.
\item Return $\f{bad}\cond{\f{high}(A)>\f{high}(B)}A{\cond{\f{high}(A)<\f{high}(B)}B{\cond{\f{count}(A)\ge\f{count}(B)}AB}}$.
\end{block}
\item \textbf{Subroutine imbalance( aligned interval $X$ ):}

\begin{block}
\item Return $\cond{\f{count}(X)\le1}{null}{\cond{\f{count}(X)>\f{count}(\f{sibling}(X))+1}X{\f{imbalance}(\f{parent}(X))}}$.
\end{block}
\item \textbf{Procedure Insert( task $t$ , window $w$ ):}

\begin{block}
\item // Create the insert state //
\item Set $w'=\f{Align}(w)$.
\item Set $(I=(n,T,W),S,r)$ to be the insert state on insertion of $(t,w')$
into $(I,S)$.
\item // Check if the instance is $2$-underallocated //
\item If $I$ is not $2$-underallocated:

\begin{block}
\item Return $Failure$.
\end{block}
\item // Insert //
\item Set $S[r]=\f{best}(W[r])$.
\item // Solve overlap //
\item If $S[i]=S[r]$ for some $i\ne r$:

\begin{block}
\item Set $S[i]=\f{best}(W[i])$.
\item Return $Success$.
\end{block}
\item // Correct imbalance //
\item Set $X=\f{imbalance}(S[r])$.
\item If $X\ne null$:

\begin{block}
\item Set $i$ such that $S[i]=\f{bad}(X)$.
\item Set $S[i]=\f{best}(W[i])$.
\end{block}
\item Return $Success$.
\end{block}
\item \textbf{Procedure Delete( task $t$ ):}

\begin{block}
\item If $t\in T$:

\begin{block}
\item Delete $t$ from $(I,S)$.
\item Return $Success$.
\end{block}
\item Otherwise:

\begin{block}
\item Return $Failure$.
\end{block}
\end{block}
\end{block}
\end{block}
\end{algo}
\begin{conjecture}[VA's properties]
\label{V.VA.prop} On an insertion of a new task, if the new aligned
instance is $2$-underallocated, VA has the following properties:
\begin{itemize}
\item If the insertion is feasible, it returns $Success$ after updating
$S$ to be a solution for the new instance by making at most $1$
reallocations and taking $O\left(\log(n)\right)$ time.
\item If the insertion is not feasible, it returns $Failure$ in $O(\log(n))$
time.
\end{itemize}
\end{conjecture}
\begin{rem*}
If \nameref{V.VA.prop} hold, it means that it is enough that the
unaligned instance remains $8$-underallocated, because by \nameref{V.Alignment}
the aligned instance would be $2$-underallocated.
\end{rem*}
\clearpage{}

\section{Multiple processors}

In this section we give a simple reduction from the multi-processor
problem to the single-processor problem as well as a bound on the
necessary slack for an efficient allocator to exist.

\subsection{Reduction to one processor\label{M.Reduction}}

For sufficiently large underallocation, the problem for $p$ processors
can be \textquoteleft solved\textquoteright{} by using any single-processor
allocator $A$ on a transformed version of the problem where all windows
have the same start time but length shortened by $1-\frac{1}{k}$
and the task length is also shortened by the same amount to $\frac{1}{k}$,
where $k=\floor{\frac{p+1}{2}}$. Each time $A$ makes a sequence
of allocations to maintain a solution $S$ in the transformed problem,
we deallocate all the corresponding tasks in the original problem,
and then allocate each one to a slot with the same start time as in
$S$ to a processor where it can actually fit. Such a processor must
exist because the number of slots in the original system that overlap
the desired slot is at most $2k-1\le p$.

Notice that this method can only utilize an odd number of processors,
otherwise it will essentially discard one processor! Also, any $\e$-slack
instance with $p$ identical windows of length $1+\e$ would under
this transformation become an instance with identical windows of length
$\e+\frac{1}{k}$, which can accommodate all the $p$ slots of length
$\frac{1}{k}$ only if $\e+\frac{1}{k}\ge\frac{p}{k}$, which implies
$\e\ge\frac{p-1}{k}\ge2\frac{p-1}{p+1}$.

On the other hand, if $p$ is odd and $I$ is the original $\e$-slack
instance when using $p$ processors, for some $\e>2\frac{p-1}{p+1}$,
then the transformed instance $J$ has an $(\e+\frac{1}{k})$-solution
$S$ using $p$ processors. Thus $J$ has an $(\e+\frac{1}{k})\frac{1}{p}$-solution
using $1$ processor, because we can allocate the tasks in order of
their corresponding slots in $S$, each to the leftmost possible position
within its corresponding slot in $S$, which is always possible because
the depth of the arrangement of slots in $S$ is at most $p$. Hence
$J$ for $1$ processor will have slack at least $(\e+\frac{1}{k})\frac{1}{p}\div\frac{1}{k}-1=\frac{\e(p+1)-2(p-1)}{2p}>0$.

This implies that any single-processor allocator that works if the
instance remains $\g$-underallocated would give a multi-processor
allocator that works if the instance remains $(2\g+1)$-underallocated.

\subsection{Inefficiency for small slack\label{M.SmallSlack}}

For $p$ processors where $p>1$, it turns out that there is no efficient
allocator for arbitrarily small slack. Specifically, given any $\e<\frac{1}{4p-1}$,
even if the instance is always $\e$-slack, it is impossible to avoid
reallocating $\W(\frac{n}{p})$ tasks per operation on average for
the following operation sequence.

Insert for each $i\in[-k..k]$ a group (indexed by $i$) of $p$ tasks
with windows $[i+\frac{j}{p},i+\frac{j}{p}+1](1+\e)$ for $j\in[1..p]$.
There is an essentially unique solution up to a relabeling of processors,
because $2(1+\e)-\frac{1+\e}{p}<2$. Call this Position~1. Now delete
the $p$ tasks in group $0$ and insert $(p-1)$ tasks with windows
$[i+\frac{j}{p}+\frac{1}{2p},i+\frac{j}{p}+\frac{1}{2p}+1](1+\e)$
for $j\in[1..p-1]$. Since $2(1+\e)-\frac{1+\e}{2p}<2$, adjacent
tasks on the same processor cannot have overlapping windows, and there
is again an essentially unique solution. Call this Position~2. We
just need to perform $(p-1)$ deletions and $p$ insertions to return
from Position~2 to Position~1. Furthermore, it is easy to check
that $kp$ reallocations are needed to get from one position to another.
Alternating in this way between the two positions forces at least
$2kp$ reallocations for every $2(2p-1)$ operations, despite the
instance remaining $\e$-slack throughout.

\clearpage{}

\section{Variable task lengths}

Up to now we have only considered unit-length tasks, in which case
it does not matter greatly whether we have a $\g$-underallocated
instance prior to insertion or after the insertion. However, this
distinction becomes important in the case of variable-length tasks,
as the two examples below will demonstrate.

For any $\g\ge1$ the following sequence of operations require $\W(n)$
reallocations per operation on average with at most $n$ tasks in
the system at any time. Insert $k$ tasks with length $1$ and window
$[1,3]k\g$. Then insert $2$ tasks with window $W_{1}=[0,2]k\g$,
the first with length $k$ and the second with length $(2k\g-k)$.
Then delete the last two tasks and reinsert them but both with window
$W_{2}=[2,4]k\g$ instead. Alternating between $W_{1}$ and $W_{2}$
would keep the instance $\g$-underallocated before each insertion
and require $k$ reallocations per alternation.

However, if we want the instance to be always $\g$-underallocated
even after insertion, then the above example does not work. But if
$\g<2$, we can show that the following operation sequence forces
$\W(n)$ reallocations per operation. Let $m=\ceil{\frac{2}{2-\g}}\ge2$
and $r=\frac{1}{m}$ and $k\in\nn$ such that $rk\in\nn$. Insert
$(1-r)k$ tasks with length $1$ and window $[0,1]k\g$. Let $V$
be an interval such that $\f{span}(V)=rk\g$ and $\f{density}(V)\ge\frac{1-r}{\g}$
now, which exists by pigeonhole principle since $rk\g\mid k\g$. Then
insert $1$ task with length $rk$ and window $V$. The instance is
still $\g$-underallocated but $\left(\frac{1-r}{\g}rk\g+rk\right)-rk\g=(2-\g-r)rk$
$\ge(2-\g-\frac{2-\g}{2})\frac{2-\g}{4-\g}k$ $=\frac{(2-\g)^{2}}{2(4-\g)}k\ge\frac{1}{6}(2-\g)^{2}k>0$,
and hence at least that number of tasks must be reallocated. Repeatedly
deleting and reinserting that task as above gives the claim. This
leaves the case of $\g\ge2$ unanswered.

In both examples, making $p$ copies of each insertion and deletion
produces a sequence that forces $\W(\frac{n}{p})$ reallocations per
operation.

\section{Open questions}

Firstly, is \nameref{V.VA.prop} (\ref{V.VA.prop}) true? If so, it
would be essentially optimal in terms of the worst case on insertion.
If not, is there some $\g$ and $m$ and an allocator that takes at
most $m$ reallocations on each insertion given that the instance
always remains $\g$-underallocated? Also, what is the minimum such
$\g$ or $m$ for the unaligned and aligned cases? Based on the examples
in this paper, we guess that $\g=\frac{4}{3}$ is minimal in the unaligned
case and $m=1$ is minimal in both cases.

Secondly, if variable-length tasks are allowed, we have shown that
it is not efficiently solvable if the lower bound on the underallocation
is less than $2$, so is there an allocator that takes $O(1)$ reallocations
on each insertion given that the instance always remains $2$-underallocated?
We again guess that such an allocator exists.

Thirdly, the multi-processor reduction only works when the instance
is more than $2\frac{p-1}{p+1}$-slack, but there is no obvious reason
why there should not be an efficient allocator if the instance remains
just $\frac{1}{4p-1}$-slack. What is the optimal allocator in that
case?

\clearpage{}

\bibliographystyle{plain}
\bibliography{RPIS}

\end{document}